\documentclass[draftclsnofoot,twoside,onecolumn,12pt]{IEEEtran}
\usepackage{amsmath}
\pdfoutput=1

\usepackage{graphicx}
\usepackage{latexsym}
\usepackage{amssymb}
\usepackage{cite}
\usepackage{color}
\usepackage{multirow}
\usepackage{algorithmic,algorithm}

\usepackage{caption}
\usepackage{comment}

\usepackage{footnote}
\usepackage{subfig}

\newtheorem{theorem}{\bf{Theorem}}
\newtheorem{lemma}{\bf{Lemma}}

\newtheorem{proof}{\bf{Proof}}
\newtheorem{corollary}{\bf{Corollary}}

\usepackage{comment}
\usepackage{cuted}
\usepackage{mathtools}
\DeclarePairedDelimiter\ceil{\lceil}{\rceil}
\DeclarePairedDelimiter\floor{\lfloor}{\rfloor}
\usepackage{derivative}
\usepackage{amsmath,amssymb} 

\begin{document}
	
	\title{Fundamental Rate-Memory Tradeoff for Coded Caching in Presence of User Inactivity} 
	\author{Jialing Liao,~\IEEEmembership{Member,~IEEE}, Olav Tirkkonen,~\IEEEmembership{Member,~IEEE}   
		\thanks{J. Liao and O. Tirkkonen are with the Department of Communications and Networking (Comnet), Aalto University, Finland (e-mail: $\rm jialing.liao@ieee.org$, $\rm olav.tirkkonen@aalto.fi$). 
		}\thanks{This work was funded by the Academy of Finland, grant 319058.}} 
	
	\maketitle
	\begin{abstract}
	    Coded caching utilizes proper file subpacketization and coded delivery to make full use of the multicast opportunities in content delivery, to alleviate file transfer load in massive content delivery scenarios. Most existing work considers deterministic environments. An important practical topic is to characterize the impact of the uncertainty from user inactivity on coded caching. We consider a one server cache-enabled network under homogeneous file and network settings in presence of user inactivity. Unlike random or probabilistic caching studied in the literature, deterministic coded caching is considered, with the objective to minimize the worst-case backhaul load by optimizing the file subpacketization and the caching strategy. First, a coded caching method is used, where each file is split into the same type of fragments labeled using sets with fixed cardinality, and the optimality of the selected cardinality is proved. Optimal file subpacketization by splitting the file into multiple types of fragments labeled with multiple cardinalities is then discussed. We show that the closed-form optimum turns out to be given by a fixed cardinality---optimizing for user inactivity only affects file delivery, cache placement is not affected. A decentralized version is also discussed and analyzed, where each user fills its storage independently at random without centralized coordination, and user inactivity is taken into account in file delivery. Simulation results show that the optimization based centralized coded caching scheme provides  performance comparable to the ideal scenario assuming full knowledge of user inactivity in the placement phase, while decentralized caching performs slightly worse against user inactivity. 
	\end{abstract}
	
	\begin{IEEEkeywords}
		Coded caching, user inactivity, linear programming.
	\end{IEEEkeywords}
	
	\section{Introduction}\label{sec_intro}
	\IEEEPARstart{E}{dge} caching is a promising technology to deal with the high demands on backhaul by bringing popular content closer to the network edge. As such, it is a promising component for future wireless cellular networks facilitated by edge intelligence \cite{M.MaddahAliMay2014, MaddahAli2014,G.Paschos2016,Yu2018}. Caching technologies have been explored from various design aspects, considering, e.g., network topology, caching model, performance metric, control structure and mathematical tool \cite{L.Li2018}. Typical cache-enabled cellular networks include cache-enabled macro-cellular networks  \cite{Khreishah2015}, heterogeneous networks (HetNets) \cite{C.Yang2016, E.Bastug2014}, device-to-device (D2D) networks \cite{B.Chen2017}, and cloud radio access networks (CRANs) / fog RANs (F-RANs) \cite{M.Tao2016, S.H.Park2016}. A caching policy is determined in a centralized or decentralized manner with performance metrics varying from backhaul load, latency, successful transmission rate, etc., utilizing tools such as optimization, stochastic geometry and deep learning. Cache placement methods considered include deterministic, random, and online placement \cite{A.Gharaibeh2016}, with or without file spitting and erasure coding based on, e.g., maximum distance separable (MDS) codes \cite{J.Liao2017}. Content delivery, on the other hand, can be realized via multiple unicast transmissions or one multicast transmission. Overall, caching models such as Femto caching \cite{K.Shanmugam2013}, probabilistic caching \cite{Blaszczyszyn2015}, and information-theoretic coded caching (CC) \cite{M.MaddahAliMay2014} have commanded attention in the literature. By making full use of the cached content in local storage and transmissions at the backhaul to create opportunities for simultaneous coded multicasting, coded caching can considerably reduce the backhaul load, i.e. the number of coded messages transmitted in parallel, with a considerable global caching gain. 
	
	Coded caching strategies have been developed under various network settings and content properties~\cite{M.MaddahAliMay2014, MaddahAli2014, G.Paschos2016, J.Zhang2015, J.Zhang2015a, S.Wang, E.Parrinello2020,S.Jin2016, MozhganBayat, M.Ji2016, Tandon2016, N.Mital2020}. In \cite{M.MaddahAliMay2014, MaddahAli2014}, centralized and decentralized coded caching methods with homogeneous network and file settings were proposed. Nonuniform file popularity, file sizes and cache sizes were considered in \cite{J.Zhang2015, J.Zhang2015a, S.Wang}, respectively. In \cite{J.Zhang2015}, file popularity was approximated into several levels and the cache space for storing files with different popularity levels was optimized. Moreover, coded caching techniques have been improved to deal with shared storage, multi-antennas, multi-requests, and the large-scale file subpacketization that follows, by jointly considering the caching design and possible collaboration. In \cite{E.Parrinello2020}, coded caching was considered with shared caches, multi-antennas and multi-requests, where index coding was used for user-to-cache association. Multi-round delivery was utilized based on the cache replication strategy proposed in \cite{S.Jin2016, MozhganBayat}. Coded caching in D2D networks where the mobile stations act as both transmitters and receivers were discussed in~\cite{M.Ji2016}. The effective numbers of sources and destinations were  modeled in caching design. \cite{Tandon2016} jointly utilized the local storage and computing to achieve edge intelligence in F-RANs.  
	
	Most of the works on coded caching are targeted for wired networks where both network and content properties are static. However, in practice, there are many stochastic properties that impose the need to optimize coded caching against uncertainty. Coded caching needs to be studied considering the impacts of randomness caused by wireless channel fading \cite{Bayat, S.P.Shariatpanahi2019, D.Cao2019}, multiple antennas and transmission interference \cite{A.Toelli2020, S.S.Bidokhti2016}, random user behavior, such as user inactivity and mobility \cite{C.Yapar2019, Ozfatura2020}, and dynamic content popularity \cite{R.Pedarsani2016}. Coded caching over a broadcast channel between the server and the users was studied in \cite{Bayat} by joint optimization of caching and modulation.  Mapping the multicast coded messages to the symbols of a signal constellation helps the users demodulate the desired symbols more reliably. In \cite{A.Toelli2020},  coded caching was applied to a multiple-input single-output (MIMO) broadcast channel by joint designing the multigroup multicast beamforming and the coded caching policy so that it benefits from spatial multiplexing, improved interference management and multi-antenna multicasting opportunities in content delivery.  
	
	An important source of uncertainty in content delivery, especially in a mobile network scenario, is user inactivity. The main reason is that the users are free to change location between cache placement and cache delivery. With edge-optimized caching, caching is local, e.g., optimized on a per base station level. Due to mobility, a user may not be within the range of the same cache at the cache delivery phase as during cache placement. Moreover, at a given tie period of cache delivery, not all users may be requesting a file.  
	When this happens, content delivery will be targeted only for active devices. As information about user activity is not available during cache placement, the design of optimal cache placement becomes complicated. The objective of this work is to find optimal cache placement and delivery strategies in the presence of user inactivity. 
	
	\subsection {Related Work}
	
	In \cite{M.MaddahAliMay2014}, Maddah-Alo and Niesen (MAN) presented pioneering information-theoretic research on coded caching, where the network topology was deterministic without user inactivity, inspiring a substantial body of scientific work. Of relevance to this paper are \cite{MaddahAli2014, N.Mital2020, C.Yapar2019, Deng2020, Daniel2020, Q.Wang2019,Dutta2021} which consider coded caching either from the perspective of optimization, or with user inactivity, or in a decentralized manner. 
	
	Multi-server networks in a random topology were considered in \cite{N.Mital2020}, where the users were randomly connected to a fixed number of servers.  Maximum distance separable (MDS) codes were utilized to construct file pieces and thus enabled the users to recover the required file with fewer fragments from a limited number of servers out of all. This is an opposite scenario of what is discussed in this paper. Here we consider one server, and the randomness is in the set of users connected to the server. 
	
	The paper \cite{C.Yapar2019} characterized user inactivity for a cache enabled D2D network with $K$ users. Each user might be inactive independently at a given probability, thus the number of effective devices can be predicted based on probability. Considering D2D, each user can both transmit and receive content from the other $K-1$ users, and hence there is a multiplier $K-1$ in file subpacketization when all users are active. When part of the users became inactive, the number of effective users available to deliver and receive content dropped to $K-1-\alpha$.
	Given an $\alpha$, cache placement and delivery performance are analyzed.  The selection of $\alpha$ is from the perspective of performance analysis. The outage probability for successful transmission was defined as a function of $\alpha$, and one was able to choose a proper $\alpha$ with any given outage probability threshold. MDS codes were utilized for multi-server transmissions, as in~\cite{C.Yapar2019}. The D2D scenario analyzed in~\cite{C.Yapar2019} is more similar to the scenario of~\cite{N.Mital2020} with fewer transmitters available than to the scenario of interest here. Instead of probabilistic performance analysis, we pursue optimization for obtaining the best file subpacketization and caching strategy.
	
	\cite{Deng2020} and \cite{Daniel2020} provided insights on optimization based coded caching design for nonuniform file parameters. User inactivity was not considered.
	
	Decentralized coded caching has been investigated based on random cache placement operated independently at each user, which removes the need for central coordination which is not always applicable for practical wireless cellular networks \cite{MaddahAli2014,Q.Wang2019,Dutta2021}. \cite{MaddahAli2014} proposed a framework for decentralized coded caching and analyzed its application in three different typologies: tree topology, shared caches and multiple requests. As an extension to \cite{MaddahAli2014}, \cite{Q.Wang2019} provided an optimization framework for decentralized coded caching in a more general scenario with arbitrary file sizes and cache sizes. Caching parameters were optimized, aiming to minimize the worst-case or average load. \cite{Dutta2021} revisited the shared caching problem where multiple users were served by one cache. An optimal delivery scheme was proposed utilizing index coding. 
	
	In multiround delivery when multiple users share a cache~\cite{E.Parrinello2020,S.Jin2016,MozhganBayat}, in some rounds not all users are present. Thus, from a cache delivery perspective, the results of~\cite{E.Parrinello2020,S.Jin2016,MozhganBayat} directly apply to an inactive user scenario. However, the {\it cache placement} setting is different. The user caching profile is assumed known during cache placement in ~\cite{E.Parrinello2020}, decentralized caching is applied in ~\cite{S.Jin2016}, while MAN cache placement with pre-selected cardinality is assumed upfront in~\cite{MozhganBayat}. In contrast, here we consider deterministic caching in a situation where the set of active users is not known during cache placement, and optimize cache placement. 
	
	In this paper, we focus on 
	a scenario where there are several cache-enabled users connected to a single server via shared links, and where there is inactivity. Centralized and decentralized coded caching are studied for cache-enabled networks with user inactivity. 
	For centralized caching, two methods are considered. First, a method with file fragments of one size is considered, and optimal fragment size is found. Second, a general scheme is considered where each file is divided into fragments of different sizes, and fragment sizes are optimized over. It is proved that the optimal solution for cache placement is the same for the basic scenario without user inactivity, and the scenario with user inactivity. 
	Mathematical analysis and simulation results are presented to illustrate the advantages of the proposed method in terms of reducing the backhaul load against user inactivity, as well as the equivalence among the subpacketization optimization in all scenarios. Finally, the decentralized coded caching strategy is of interest in scenarios with uncertainty resulting from user inactivity.
	Analysis and simulations are provided to compare decentralized coded caching with centralized coded caching in the presence of user inactivity. 
	
	\subsection {Contributions}
	In this paper, our aim is to unlock the potential of utilizing coded caching against the uncertainty caused by user inactivity. In summary, this paper has made the following major contributions: 
	\begin{itemize}
		\item We address coded caching design in presence of user inactivity using both centralized and decentralized cache placement. The uncertainty of the inactive users causes some difficulties for caching placement design. 
		
		\item We develop centralized coded caching schemes optimizing the worst-case backhaul load of the one server shared link network via file subpacketization optimization assuming fixed cardinality and also multiple cardinalities.   
		
		\item With fixed cardinality of the fragment label set, the optimal cardinality is proved to be the same as the cache replication parameter used in Maddah-Ali-Niesen's method without user inactivity. 
		\item Considering the possible redundancy introduced by the fragments labeled with the user sets containing inactive users,  file subpacketization is also done based on multiple cardinalities. The weights for different types of fragments labeled with different cardinalities are designed. The optimal solution is proved to be the same as fixed cardinality. 
		
		\item We have utilized decentralized cache placement in a system with user inactivity and developed inactivity-aware cache delivery for decentralized caching. 
		
		\item The performance gaps in terms of backhaul load against user inactivity have been investigated between the proposed centralized method and the decentralized method, and also between the centralized method and the ideal MAN method. While the former decreases with the increase of the number of inactive users, the latter increases from $0$ to a peak point with the number of inactive users, and then go down to $0$ until all the users become inactive.     
		
		\item Simulations have shown that the proposed optimization based coded caching method shows compatible performance to the ideal scenario with full user inactivity information available in the placement phase. The decentralized method with user inactivity is easy to implement at a price of slight performance degradation. 
		
	\end{itemize}

	\subsection {Notation} 
	The notation $[b]$ denotes the set consisting of consecutive integers $\{1, 2, {\dots }, b-1, b\}$. Similarly, $[a:b]$ is used to define the set $\{a, a+1, {\dots }, b-1, b\}$ consisting of integers ranging within $[a,b]$. $W_n$ is used to refer to the $n$-th file with $|W_n|$ denoting the length of the file. A fragment of file $W_n$ is expressed as $W_{n, \mathcal \tau}$, stating that the fragment of file $n$ is stored at the users whose indices belong to the set $\mathcal \tau$. The cardinality of any set $\mathcal \tau$, i.e. the number of elements in set $\mathcal \tau$, is denoted by $|\mathcal \tau|$. For any real number $c$, $\lfloor c \rfloor$ and $\lceil c \rceil$ denote the floor and ceiling versions of $c$, respectively. The operator $\oplus$ denotes the bitwise “XOR” operation between multiple fragments.

	\subsection {Organization} 
	The organization of the rest of the paper is as follows. Section II introduces the system model used in the paper and  Section III discusses the coded caching scheme in presence of user inactivity with fixed cardinality in file subpacketization. Section IV introduces an optimization framework and its optimal solution to interpret coded caching against user inactivity based on multiple cardinalities in file subpacketization. Section V discusses the alternative decentralized coded caching scheme to deal with the uncertainty caused by user inactivity which is suboptimal but easy to implement. Section VI provides a comparison of centralized and decentralized manners, as well as the ideal MAN method, against user inactivity. Section VII presents the simulation results of the proposed coded caching schemes against user inactivity. Section VIII summarizes the paper with a conclusion and discussion of the main contributions.

	\section{System Model}\label{sec_sys}
	In this section, the network model with caching policies, as well as the content characteristics that involve the structure of the network coding, and the file popularity profiles are presented. 
	There is a base station connected to the core network with access to all the file library ($N$ files $W_1, W_2,\dots, W_N$ each with equal size $F$ and popularity), and $K$ users each with local storage of size $MF$. The users are connected to the server via error free shared links. The probability for each user to be inactive is $p$. In the cache placement phase, the user inactivity is unknown while the base station has the information of the user inactivity in the delivery phase. Assuming that in a realization, there are $I$ inactive users forming an inactive user set ${\cal I}$. To ensure the significance of the discussion, we assume there is at least one user becoming inactive and at the same time, there is at least one active user to be served, i.e. $I \in [K-1]$. The number of active users is correspondingly defined as $J=K-I$. The probability for $I$ of the $K$ users being inactive is 
	\begin{equation}
		P(I)={K \choose I} p^I (1-p)^{K-I},~I \in [K-1]. 
	\end{equation}
	
	The cached content at the local cache of a user $k$ is defined as $Z_k$. The content delivered through the backhaul via coded multicast is described as $X_{\boldsymbol{d}}$, where  $\boldsymbol{d}=(d_1,d_2,\dots,d_K)$ with $d_k \in [N], k \in [K]$ denoting the demand of user $k$.   
	
	As a common metric for measuring the performance of coded caching methods, the backhaul load is defined as the volume of content needed to be delivered via backhaul using coded multicasting. The backhaul load can be calculated both in the worst case when the active users each requesting a different file, and the average case with all types of possible demands considered. Here, the worst-case backhaul load is considered which implies that the number of files is higher than the number of users $N>K$. We aim to minimize the worst-case backhaul load by designing the caching strategy subject to file size and cache size constraints. \footnote{Unless otherwise specified, the backhaul load in the following parts of the paper refers to the worst-case backhaul load.} 
	
	\begin{figure}[htb]
		\centering
		\includegraphics[width=13cm]{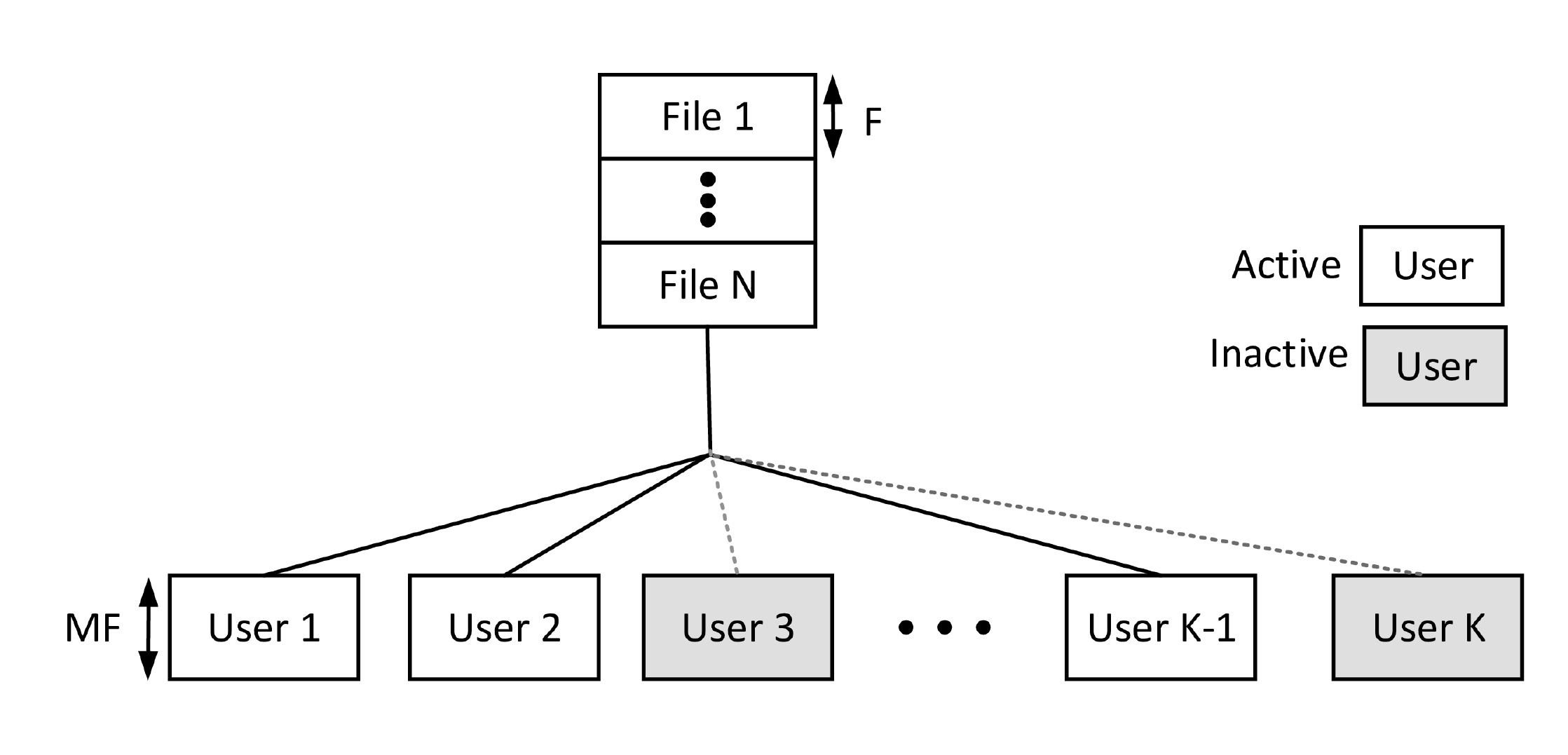}
		\caption{System model for a cache-aided network in presence of user inactivity.}\label{fig_1}  
	\end{figure}
	
	\section{coded caching in presence of inactive users}\label{method1}  
	\subsection{Content Placement and Delivery with User Inactivity} 
	We begin with the effective file subpacketization in Maddah-Ali-Niesen's method to make full usage of the multicast delivery opportunities. In the MAN method, all the users have the same cache content placement and all the files are equally cached in local storage because of the homogeneous settings. Define a variable $t \triangleq \frac{KM}{N}$, and then divide each file into ${K \choose t}$ fragments equally. $t$ is referred to as the cache replication parameter in literature \cite{ S.Wang}. \footnote{Here MAN method refers to Maddah-Ali-Niesen's method. It is assumed $t$ is an integer. If not, content sharing can be used to deal with this issue.} The fragments are indexed by all subsets of users  ${\cal \tau} \subset [K]$ of fixed cardinality $|\tau|=t$. Accordingly the fragments of file $n$ are $W_{n, {\cal \tau}}$. 
	For sake of simplicity, the set of all $t$-element subsets of $[K]$ is defined as ${\cal \zeta}=\{{\cal \tau}|{\cal \tau} \subset [K], |{\cal \tau}|=t \}$. It is assumed that user $k$ stores the fragments $W_{n, {\cal \tau}}$ of each file $n$ when $k \in {\cal \tau}, {\cal \tau}  \in {\cal \zeta}$. Hence, the cache content placement at user $k$ can be written as 
	
	\begin{equation}\label{placement}
		Z_k=(W_{n, {\cal \tau}}:{\cal \tau} \in {\cal \zeta}, k \in {\cal \tau}, n \in [N]). 
	\end{equation}
	
	In each cache, there are ${K-1 \choose t-1}$ fragments for each file, and each fragment has normalized size of $1/{K \choose t}$. Thus the cache capacity constraint holds as follows 
	\begin{equation}
		N{K-1 \choose t-1}\frac{1}{{K \choose t}}=M.
	\end{equation}
	
	Without user inactivity ($I=0$), the server can deliver a number of packets  each of which comprises coded fragments to the users to help them reconstruct their requested files:
	\begin{align} \label{X_or}
		X_{\boldsymbol{d}}=(X_{\boldsymbol{d}, {\cal S}}:{\cal S} \in {\cal \vartheta}), \\
		X_{\boldsymbol{d}, {\cal S}}=\oplus_{k \in {\cal S}}~W_{d_k, {\cal S}\setminus\{k\}}, 
	\end{align}
	where the set ${\cal S}$ has one more element than set ${\cal \tau}$ satisfying ${\cal S} \subset [K], |{\cal S}|=t+1$. The set of all ${\cal S}$ is ${\cal \vartheta}$. This coded multicast strategy works in the way that all the users are able to recover their request files using the same transmitted packets  and the cached fragments in local cache. For any user  $i$ in a particular ${\cal S}$, the linear combination $X_{\boldsymbol{d}, {\cal S}}$ can be rewritten as 
	\begin{equation}
		X_{\boldsymbol{d}, {\cal S}}=W_{d_i, {\cal S} \setminus \{i\}} \oplus (\oplus_{k \in {\cal S}, k \neq i }~W_{d_k, {\cal S} \setminus\{k\}}), 
	\end{equation}
	where $W_{d_i, {\cal S} \setminus \{i\}}$ is one of the fragments that user $i$ needs to recover the requested file $d_i$. The other fragments in linear combination $(\oplus_{k \in {\cal S}, k \neq i }~W_{d_k, {\cal S} \setminus\{k\}})$ are all cached at user $i$ according to the cache placement in \eqref{placement} due to the fact that $i \in {\cal S} \setminus\{k\}$ for any $k \in {\cal S}$ but $k \neq i$. Therefore, $W_{d_i, {\cal S} \setminus \{i\}}$ can be decoded by user $i$. Taking all types of ${\cal S} \in {\cal \vartheta}, i \in {\cal S}$ into account, user $i$ is thus able to decode all the missing fragments of the request file, i.e. $(W_{d_i, {\cal S} \setminus \{i\}}:{\cal S} \in {\cal \vartheta}$, $i \in {\cal S})=(W_{d_i, {\cal \tau}}:{\cal \tau} \in {\cal \zeta}$, $i \not\in {\cal \tau})$.  
	
	As the linear combination of several fragments via operator $\oplus$ has the same size as a single fragment, i.e. $F/{K \choose t}$, the backhaul load can be written as the size of a fragment multiplied by the number of the different sets ${\cal S}$ as follows:
	\begin{equation}\label{R}
		R=\frac{F}{{K \choose t}} {K \choose t+1}=F~\frac{K-t}{t+1}. 
	\end{equation} 
	Because file size $F$ performs as a multiplier in backhaul loads, unit file size is assumed in the following for briefness. 
	
	
	In~\cite{E.Parrinello2020}, multi-round delivery to users sharing caches is considered. In a given delivery round, a subset of cache profiles may be present. Thus, for a given round, the cache delivery problem addressed in \cite{E.Parrinello2020} is the same as the delivery problem in a situation with a set of inactive users. We shall thus use the cache delivery scheme of  \cite{E.Parrinello2020}, rephrased to an inactive user scenario.

		Assuming there are $I$ inactive users forming an inactive user set ${\cal I} \subset [K]$, we utilize a general cardinality $l \in [K]$ in file subpacketization instead of the cache replication parameter $t=KM/N$ that is used in the MAN method. The optimal value of $l$ will be optimized in Subsection \ref{method2}. The transmitted packets would be 
		\begin{align} \label{Xd_inact}
			    X_{\boldsymbol{d}}=\left\{\begin{array}{cl}
				(X_{\boldsymbol{d}, {\cal S}}:{\cal S} \in {\cal \vartheta}), & \mbox{if} ~l+1 > I,  \\
				\\
				(X_{\boldsymbol{d}, {\cal S}}:{\cal S} \in {\cal \vartheta},{\cal S} \not\subset {\cal I}), &  \mbox{if}~l+1 \le I,
			\end{array}\right. 
		\end{align} 
		where the packet given any subset ${\cal S}$ is   
		\begin{equation}\label{Xds_inact}
			X_{\boldsymbol{d}, {\cal S}}=\oplus_{k \in {\cal S}, k \notin {\cal I}}~W_{d_k, {\cal S}\setminus k}.
		\end{equation} 
		The worst case backhaul load then becomes~\cite{E.Parrinello2020}\footnote{${K \choose l+1}/{K \choose l}=-1+\frac{K+1}{l+1}$ is used for simplifying the computation.} 
		\begin{align}\label{R_inact-obj}
			R(l)=\left\{\begin{array}{cl}
				\frac{1}{{K \choose l}} {K \choose l+1}, & \mbox{if}~l+1> I,\\
				\frac{1}{{K \choose l}} \left[{K \choose l+1}-{I \choose l+1}\right], &\mbox{if}~l+1\le I. 
			\end{array}\right. 
		\end{align}

	Comparing \eqref{R} and \eqref{R_inact-obj}, it is obvious that the worst case backhaul load in presence of user inactivity is either the same as the one derived without user inactivity when $t+1 > I$ or is smaller than the one without user inactivity when $t+1 \le I$. 
	For clarification, we summarize the procedure of the proposed centralized coded caching scheme in presence of user inactivity in Alg.~\ref{alg_cen}, based on (\ref{R_inact-obj}). Note that in Alg.~\ref{alg_cen}, $l^*$ denotes the optimal cardinality based on file subpacketization optimization given by $l^*=KM/N$, which shall be carefully proved in Subsection \ref{method2} and Section \ref{method3}. 
	
		\begin{algorithm} 
		\caption{Centralized Coded Caching in Presence of User Inactivity}
		\label{alg_cen}
		\begin{algorithmic}[1]
			\STATE \textbf{procedure} {\large P}LACEMENT
			\STATE ~~~~$l \leftarrow l^*$ (the optimal solution in file subpacketization optimization: $l^*=KM/N$)
			\STATE ~~~~${\cal \zeta} \leftarrow \{{\cal \tau}|{\cal \tau} \subset [K], |{\cal \tau}|=l \}$
			\STATE ~~~~\textbf{for} $n \in [N]$ do 
			\STATE~~~~~~~split $W_n$ into $(W_{n, {\cal \tau}}|{\cal \tau} \in {\cal \zeta})$ with identical size
			\STATE~~~~\textbf{end for}
			\STATE ~~~~\textbf{for} $k \in [K]$ do 
			\STATE~~~~~~~user $k$ caches $Z_k \leftarrow (W_{n, {\cal \tau}}|{\cal \tau} \in {\cal \zeta}, k \in {\cal \tau},n \in [N])$ 
			\STATE~~~~\textbf{end for}
			\STATE \textbf{end procedure} \\
			Users make requests $\boldsymbol{d}$ given the number and identity of the inactive users $(I, {\cal I})$
			\STATE \textbf{procedure} {\large D}ELIVERY	
			\STATE ~~~~$l \leftarrow KM/N$
			\STATE ~~~~${\cal \vartheta} \leftarrow \{{\cal S}|{\cal S} \subset [K], |{\cal S}|=l+1\}$
			\STATE ~~~~\textbf{if} $l+1>I$ do\\
			\STATE~~~~~~~$X_{\boldsymbol{d}} \leftarrow (\oplus_{k \in {\cal S}, k \notin {\cal I}}~W_{d_k, {\cal S}\setminus \{k\}}:{\cal S} \in {\cal \vartheta})$
			\STATE ~~~~\textbf{else if} $l+1 \le I$ do\\
			\STATE~~~~~~~$X_{\boldsymbol{d}} \leftarrow (\oplus_{k \in {\cal S}, k \notin {\cal I}}~W_{d_k, {\cal S}\setminus\{k\}}:{\cal S} \in {\cal \vartheta},{\cal S} \not\subset {\cal I})$
			\STATE~~~~\textbf{end if}  
			\STATE \textbf{end procedure}
		\end{algorithmic}
	\end{algorithm}
	While in (\ref{R_inact-obj}) and Alg.~\ref{alg_cen} a framework of coded caching in presence of user inactivity is presented, there is an essential problem remaining:  {\it What is the optimal cache placement policy if it is known at the cache placement that a random set of $I$ out of $K$ users will be inactive at the time of cache delivery?} This motivates our investigation on file subpacketization optimization in the following.  
	
	\subsection{Optimizing Coded Caching in Presence of User Inactivity}
When optimizing cache placement in presence of user inactivity, a remark is in place.  It is fair to assume that the server has no information of user inactivity in the content placement phase while this type of knowledge becomes available in the content delivery phase. As a result, the server will only target the active users for content delivery based on a given cache content placement. The coded packets to be transmitted from the server are only for $J=K-I$ active users. Thus the number of inactive users $I$ affects at least content delivery and backhaul load. Moreover, to determine the packets $Z_{\boldsymbol{d}}$ to be transmitted, full information about file requests is needed, which directly indicates the set inactive users ${\cal I}$. As backhaul load, either worst case or average, is usually selected as a performance metric in coded caching design, $I$ has to appear in the objective of cache placement optimization, not only in cache delivery. This conflicts with the assumption that user inactivity information is not available before content delivery. For the sake of analysis, we assume that the server knows the estimated number of inactive users $I$ already in the cache placement phase, while not knowing which set of users ${\cal I}$  will become inactive. This information may be based, e.g., on historical information.
It turns out that the optimal solutions to file subpacketization problems both with fixed cardinality and with multiple cardinalities, will be independent on $I$.  
This demonstrates that in presence of user inactivity, where the set ${\cal I}$ is not known in the cache placement phase, knowledge about the cardinality of ${\cal I}$ is of no use. The optimal schemes found here are thus optimal also in situations where $I$ is not known during cache placement.

	\subsection{Subpacketization Optimization with Fixed Cardinality}\label{method2}
	
	The subpacketization method and cache content placement used here is based on a group of subsets ${\cal \tau} \subset [K]$ with $|{\cal \tau}|=l$ to create possible multicast opportunities in the content delivery phase. The optimal choice for the cardinality of ${\cal \tau}$ can be interpreted from an optimization perspective. 
	
	Firstly, we consider the normal case without user inactivity, and define an multiple cardinalities $l=|{\cal \tau}|, l \in [K]$, to replace the fixed $t=KM/N$. The optimal $l$ should give the lowest backhaul load $R(l)={K \choose l+1}/{K \choose l}=\frac{K-l}{l+1}$ while satisfying cache capacity constraint ${K-1 \choose l-1}/{K \choose l}=l/K \le M/N $. The optimization problem can be rewritten into
	\begin{subequations}\label{opt-man} 
		\begin{align}
			\min_l&~~~~-1+\frac{K+1}{l+1} \label{optMAN_0}\\
			{\rm s.t.} &~~l \le \frac{KM}{N}, l \in [K].\label{optMAN_1}
		\end{align}
	\end{subequations} 
	Since the objective \eqref{optMAN_0} decreases with respect to $l$, the optimal solution is the largest $l$ with cache capacity constraint satisfied with equality, i.e. $l=\frac{KM}{N}$, which agrees with the cache replication parameter $t=(KM)/N$ used in the MAN method. In particular, when $(KM)/N$ is not integer, the optimal cardinality becomes $l=\floor{KM/N}$, which works for the scenario with user inactivity as well.  
	
	Similarly, we substitute the worst case backhaul load with user inactivity \eqref{R_inact-obj} into the objective function and again replace the fixed $t=KM/N$ with a variable $l$  to be optimized. The coded caching optimization with user inactivity after simplification can be written as 
	\begin{subequations}\label{optFix} 
		\begin{align}
			\min_l&~~~~R(l)  \label{optFix_0} \\
			{\rm s.t.} &~~l \le \frac{KM}{N},~~l \in [K-1]. \label{optFix_1}
		\end{align}
	\end{subequations} 
	where the objective function is given by \eqref{R_inact-obj}.
	
	To find the optimal $l$, the analysis of (\ref{optFix}) can be divided into two parts. We have 
	
	\begin{lemma}\label{lemma-optFix-obj} 
		The backhaul load $R(l)$ of (\ref{R_inact-obj}) is a decreasing function of $l$ in the interval $l \in [I-1]$. 
	\end{lemma} 
	\begin{proof} 
		See Appendix B. 
	\end{proof} 
	
	We can now show 
	\begin{theorem}\label{theorem-inactive-t} 
		If  MAN cache placement with cardinality $l$ of all caching subsets $\tau$ is used in the presence of user inactivity, minimum worst case backhaul load is achieved with cardinality  $l=\frac{KM}{N}$. 
	\end{theorem}
	
	\begin{proof}
		We first treat separately the minimization in the regions $l>I-1$ and $l\leq I-1$ separately. 
		In the region $l > I-1$, the backhaul load \eqref{R_inact-obj} is the same as the one without user inactivity \eqref{R}, for which the optimal cardinality has been proved to be $l=MK/N$. 
		Thus if $I \le KM/N$, this yields the minimum backhaul in this region, while if $I > KM/N$, all of this region is infeasible. 
		
		According to Lemma \ref{lemma-optFix-obj}, the minimum backhaul in the second region $l\leq I-1$ is achieved at the maximal feasible point $l=\min(KM/N,I-1)$. 
		
		It remains to find the smaller value of the solutions in the two regions, in the case $KM/N \geq I$. For this, we compute 
		\begin{eqnarray}
			R(I-1)-R\left(\frac{KM}{N}\right) &=&
			\frac{{K \choose I}-{I \choose I}}{{K \choose I-1}} -\frac{{K \choose KM/N +1}}{{K \choose KM/N}} \cr
			&=&\frac{(I-1)!}{I (KM/N+1) K!} \times B \, \Gamma, 
		\end{eqnarray}
		where $B=I(KM/N+1)(K-I+1)!$ and 
		\begin{eqnarray}
			\Gamma+1 &=& \frac{(K+1)(KM/N+1-I)}{I(KM/N+1)(K-I+1)! {K \choose I-1}}  \cr
			&\ge& \frac{K+1}{KM/N+1} > 1.
			\label{fixProof}
		\end{eqnarray}
		This completes the proof.
	\end{proof}
	
	Theorem~\ref{theorem-inactive-t} thus states that the optimal cardinality in file subpacketization for coded caching in presence of user inactivity in the whole interval $I \in [K-1]$ is always $l=KM/N$, which is the same as $t$ used in MAN method without user inactivity. 

	\section{Subpacketization with Multiple Cardinalities}\label{method3} 
	
	The analysis in previous section contains redundancy in the content placement caused by caching the fragments related to the inactive users. For instance, the fragments $(W_{n, {\cal \tau}}:k \in {\cal \tau},~{\cal \tau}\cap{\cal I} \neq \emptyset, n \in [N])$ stored in an active user $k$ seem to take up storage space without contributing in reducing backhaul load. The optimal file subpacketization and cache placement is to cache only the fragments corresponding to the active users, e.g. $(W_{n, {\cal \tau}}:k \in {\cal \tau}, {\cal \tau} \cap {\cal I} = \emptyset, n \in [N])$. However, the information about user inactivity is unknown in the cache placement phase, which means that the set ${\cal I}$ can not be specified. 
	
	Given an inactivity probability, the probability of a user caching a file fragment in vain grows with fragment label cardinality  $|\tau|$. Above, we found that the optimal cardinality is given by $l=KM/N$ if all labels have the same cardinality. As the number of active users decreases, there is a possibility that having labels of multiple cardinalities might lead to more efficient use of the caches. In \cite{Daniel2020}, the multiple cardinalities based file subpacketization is utilized to deal with the heterogeneous of file popularity which imposes multilevel file subpacketization in terms of popularity.

	For this, we split each file based on subsets with a series of different cardinalities $l \in [0:K]$ instead of a fixed number $t$. That is to say, each file is split into $2^K$ fragments labeled with $W_{n, {\cal A}^l}:{\cal A}^l \subset [K], |{\cal A}^l|=l,  l \in [0:K]$. Similarly, we assume that in the cache placement phase, a fragment is cached by user $k$ if its fragment label ${\cal A}^l, l \in [0:K]$ includes $k$:
	\begin{equation}
		Z_k=(W_{n, {\cal A}^l}:k \in {\cal A}^l, {\cal A}^l \subset [K],|{\cal A}^l|=l, l \in [0:K], n \in [N]). 
	\end{equation}
	According to the cardinality $l$, the fragments for each file $n$ can be divided into $K+1$ groups as $W_n^l=(W_{n, {\cal A}^l}:{\cal A}^l \subset [K], |{\cal A}^l|=l), l=0,1, \dots,K.$ There are ${K \choose l}$ types of fragments in fragment group $W_n^l$. In total, there are $\sum_l {K \choose l}=2^K$ different fragments for each file. By adjusting the weights of the fragment groups $W_n^l, l \in [0:K]$ for each file, the space that each fragment group takes from the cache is decided accordingly. The number of effective users involved in the caching design can then be controlled to some degree. 
	
	It is assumed that the fragments in the same group $l$ have the same size. Define a weight vector as ${\boldsymbol \alpha}\triangleq [\alpha^0,\alpha^1,\dots, \alpha^K]$ with $\alpha^l$ denoting the size of a fragment in fragment group $l$ normalized by file size $F$, i.e. $\alpha^l=|W_{n, {\cal A}^l}|, n \in [N]$. Hence, the size of fragment group $l$ of file $n$, $W_n^l$, is ${K \choose l} \alpha^l F$.         
	
	Now the file size and cache capacity constraints are:
	\begin{subequations} \label{c_fs} 
		\begin{align}
			&~~~~\sum_{l=0}^K {K \choose l} \alpha^l=1,  \label{c_fs0}  \\
			&\sum_{l=1}^K \alpha^l {K-1 \choose l-1} \le \frac {M}{N},  \label{c_fs1}  \\
			&~0 \le \alpha^l,~l=0,1\dots,K. \label{c_fs2} 
		\end{align}
	\end{subequations}
	In the case, the content to be delivered to the users via backhaul can be derived as (with referring to Appendix A)
	\begin{align*} 
		X_{\boldsymbol{d}}=\left\{\begin{array}{cl}
		(X_{\boldsymbol{d},{\cal A}^{l+1}}:{\cal A}^{l+1} \subset [K], |{\cal A}^{l+1}|=l+1), ~~\text{if}~l+1>I, \\
			(X_{\boldsymbol{d},{\cal A}^{l+1}}:{\cal A}^{l+1} \not \subset {\cal I},{\cal A}^{l+1} \subset [K], |{\cal A}^{l+1}|=l+1), \\
			 \quad\quad\quad\quad\quad\quad\quad\quad \quad\quad\quad\quad\quad\quad 
			 \text{if}~1 < l+1 \le I,\\
			 \quad~X_{\boldsymbol{d},{\cal A}^{l}},\quad\quad\quad\quad\quad\quad \text{if}~l=0,  
		\end{array}\right. 
	\end{align*}
where the packet corresponding to a given fragment label set ${\cal A}^{l+1}$ is given by 	
	\begin{align*} 
		X_{\boldsymbol{d},{\cal A}^{l+1}}=\left\{\begin{array}{cl}
			\oplus_{k \in {\cal A}^{l+1}, k \notin {\cal I}}~~W_{d_k, {\cal A}^{l+1}\setminus{\{k\}}} ,&\text{if}~l+1>I, \\
			\oplus_{k \in {\cal A}^{l+1}, k \notin {\cal I}}~~W_{d_k, {\cal A}^{l+1}\setminus{\{k\}}},&\text{if}~1 < l+1\le I,\\ 
			~~W_{d_k, {\cal A}^l},&\text{if}~l=0.   
		\end{array}\right. 
	\end{align*}

	In particular, there is an exceptional case for $l=0$ when ${\cal A}^0$ equals to $\emptyset$ and thus none of the users has stored the subfiles $W_{n,{\cal A}^0}, n \in [N]$. Accordingly, the backhaul load normalized by file size $F$ is written as
	\begin{align}
		\quad\quad R({\boldsymbol \alpha})&=(K-I)\alpha^0+\sum_{l=I}^{K-1} {K \choose l+1} \alpha^l
		  + \sum_{l=1}^{I-1} \bigg[{K \choose l+1} - {I \choose l+1} \bigg]\alpha^l \notag\\
		&=\sum_{l=0}^{K-1} {K \choose l+1} \alpha^l-\sum_{l=0}^{I-1} {I \choose l+1} \alpha^l.  
	\end{align} 
	The caching design turns to solving a linear programming of minimizing the backhaul load subject to the file size and cache capacity constraints \eqref{c_fs}: 
	
	\begin{equation}\label{opt0} 
		\min_{{\boldsymbol \alpha}}~~~~  R({\boldsymbol \alpha})~~~~ 
		{\rm s.t.}~~\eqref{c_fs0}-\eqref{c_fs2}.  
	\end{equation}
	
	Existing solvers, e.g. CVX \cite{M.Grant2013}, can be used to solve problem \eqref{opt0} with $K+1$ variables and $K+3$ constraints \cite{M.Tao2019,M.Tao2017}. It is important to figure out the structure of the optimal solution which is discussed below.

	To simplify the problem, we replace the original variables with the variables satisfying $\beta^l={K \choose l} \alpha^l, l \in [0:K]$. Thus we derive a new weight vector as ${\boldsymbol \beta} \triangleq [\beta^0,\beta^1,\dots,\beta^K]$. In this case, problem \eqref{opt0} can then be rewritten into   
	\begin{subequations}\label{opt1} 
		\begin{align}  
			\min_{\{\beta^l\}} ~~~~&\sum_{l=0}^{K-1} \frac{K-l}{l+1} \beta^l-\sum_{l=0}^{I-1} {I \choose l+1}/{K \choose l}\beta^l \label{R1_0}\\
			{\rm s.t.} ~~~~~~&\sum_{l=0}^K \beta^l=1, \label{R1_1}   \\
			~~~~~~~~~~&\sum_{l=1}^K l \beta^l \le t, \label{R1_2}   \\
			~~~~~~&0 \le \beta^l,~l \in [0:K]. \label{R1_3} 
		\end{align}
	\end{subequations} 
	
	It can be observed that the terms related to the inactive users in the objective actually destroy the similarity among the combination terms in the objective, and thus it becomes challenging to find a closed form solution to problem \eqref{opt1}. 
	
	To proceed, we analyze the properties of the objective function and the linear constraints. The discussion is generally accomplished in \textit{four} steps each of which is formulated as a Lemma given below, discussing \textit{the objective, the constraints, the structure of the optimal solution and an exceptional case.}
	
	\begin{lemma}\label{lemma-step1} 
		The first derivative of coefficients in the objective function \eqref{R1_0} is negative while the second derivative is positive when $I \in [K-2]$ and equal to $0$ when $I=K-1$. 
	\end{lemma} 
	\begin{proof}
		See Appendix C. 
	\end{proof}
	
	\begin{lemma}\label{lemma-step2} 
		The optimal solution to problem \eqref{opt1} must have a tight cache capacity constraint \eqref{R1_2}. 
	\end{lemma} 
	\begin{proof}
		We aim to prove the cache constraint \eqref{R1_2} must be satisfied with equality by the optimal solution $\boldsymbol{\beta}$ of problem~\eqref{opt1} via contradiction. If not, one can always derive a new feasible point, which satisfies all the constraints with a lower objective value, by assigning a larger value to $\beta^l$ with higher $l$. That is to say, the optimal caching strategy makes full usage of the cache space.
		
		We define a feasible solution as $\boldsymbol{\beta}=[\beta^0,\beta^1,\dots,\beta^K]$ guaranteeing a loose constraint \eqref{R1_2}, i.e. $t-\sum^K_{l=0} l \beta^l=\lambda>0$ with $\lambda >0$. We reduce the value of a nonzero $\beta^{l_1}$ by some value $\delta$, where $0 < \delta \le \beta^{l_1}$, then increase the value of any $\beta^{l_2}$ with $l_2>l_1$ by $\delta$. Since it holds true that $l_2>l_1$ and the coefficients in constraint \eqref{R1_2} increase linearly with $l$, we can always adjust the value of $\delta$, such that constraint \eqref{R1_2} is satisfied with a narrower gap $\lambda$. To be exact, we let $\delta$ satisfy the following constraints:  
		\begin{align}
			\left\{\begin{array}{cl}	
				0 < \delta \le \beta^{l_1},\\
				\beta^{l_2}+\delta \le 1,\\
				t-\lambda+l_1 (-\delta)+l_2 \delta \le t,
			\end{array} \right.  
		\end{align}
		which results in $0 < \delta=\min\{\beta^{l_1}, 1-\beta^{l_2}, \lambda/(l_2-l_1)\}$. The analysis is based on the assumption that there are at least two nonzero variables in $\boldsymbol{\beta}$, i.e. $0 <\beta^{l_1} <1$ and $0 <\beta^{l_2} <1$. In case that there is only one nonzero element in $\boldsymbol{\beta}$, i.e. $\beta^{l_1}=1$, we can recall the conclusion in Theorem~\ref{theorem-inactive-t} assuming fixed carnality that the optimal solution is $\boldsymbol{\beta}=\{\beta^t=1, \beta^l=0, \text{when}~l \neq t\}$ with the cache capacity constraint satisfied with equality. Hence, we only need to discuss on the general case when there are at least two nonzero elements in $\boldsymbol{\beta}$.
		
		In this case, the summation of the elements in $\boldsymbol{\beta}$ remains the same so that constraint \eqref{R1_1} still holds true. The new feasible point becomes $\boldsymbol{\hat{\beta}}=[\hat{\beta}^0,\hat{\beta}^1,\dots,\hat{\beta}^K]$ with
		\begin{align}
			\left\{\!\begin{array}{cl}
				&\hat{\beta}^{l_1}=\beta^{l_1}-\delta, \\ &\hat{\beta}^{l_2}=\beta^{l_2}+\delta, \\ &\hat{\beta}^l=\beta^l, \mbox{else}.
			\end{array} \right.
		\end{align} 	 
		Since only two variables are changed, the difference of the objective function defined as $\Delta$ can be written as 
		\begin{align}
			\Delta=&(c^{l_1} \hat{\beta}^{l_1}+c^{l_2} \hat{\beta}^{l_2})-(c^{l_1} \beta^{l_1}+c^{l_2} \beta^{l_2}) \notag \\
			=&c^{l_2} \delta-c^{l_1} \delta=(c^{l_2}-c^{l_1})\delta \overset{(a)}{<} 0, \label{tightC}
		\end{align} 
		where $(a)$ is derived utilizing the conclusion in previous step that coefficient $c^{l}$ is decreasing with respect to $l$. As is presented in \eqref{tightC},  the new feasible point $\boldsymbol{\hat{\beta}}$ gives a lower objective value than $\boldsymbol{\beta}$. That is to say there always exists a better solution $\boldsymbol{\hat{\beta}}$. It is consequently proved that the optimal solution to problem \eqref{opt1} always satisfies the cache capacity constraint \eqref{R1_2} with equality. 
	\end{proof}

	\begin{lemma}\label{lemma-step3} 
		The optimal solution has at most one non-zero variables of $\boldsymbol{\beta}$ (two for non-integer $t$). 
	\end{lemma} 
	\begin{proof}
		Here contradiction will again be employed to prove that the optimal solution $\boldsymbol{\beta}$ must have either one non-zero variable ($t \in \mathbb{Z}$), or two non-zero variables with consecutive indices ($t \not\in \mathbb{Z}$), i.e. some $l$ and $l+1$. 
		
		From the contradictory perspective, we assume a feasible solution $\boldsymbol{\beta}$ to problem \eqref{opt1}, which has at least two non-zero variables defined as $\beta^{l_1}$ and $\beta^{l_2}$ with $l_2-l_1 \ge 2$, is optimal, and then prove that there exists a better solution $\boldsymbol{\beta_o}$. Similar to the construction of $\boldsymbol{\hat{\beta}}$ in previous step, we reduce a small positive value $\delta$ from one $\beta^l$ and add it to another $\beta^l$ to keep the sum unchanged. Thus $\boldsymbol{\beta_o}=[\beta^0_o, \dots \beta^{K}_o]$ is defined as 
		\begin{align}
			\left\{\!\!\begin{array}{cl}	
				&\beta^{l_1}_o=\beta^{l_1}-\delta,\\
				&\beta^{l_1+1}_o=\beta^{l_1+1}+\delta,\\
				&\beta^{l_2-1}_o=\beta^{l_2-1}+\delta, \\
				&\beta^{l_2}_o=\beta^{l_2}-\delta, \\
				&\beta^l_o=\beta^l,~\mbox{else}.
			\end{array} \right.  
		\end{align}
		Specially, we set $\beta^{l_2-1}_o=\beta^{l_2-1}+2\delta$ when $l_2-l_1=2$ instead. The cache capacity constraint is satisfied with equality because 
		\begin{equation*}  
			\sum_l l \beta^l_o=\sum_l l \beta^l-l_{1}\delta+(l_{1}+1)\delta+(l_{2}-1)\delta-l_{2}\delta=t,
		\end{equation*}
		where $\sum_l l \beta^l=t$ proved in previous step. 
		
		To proceed, one needs to derive the change of objective value with $\boldsymbol{\beta_o}$ from the one with $\boldsymbol{\beta}$. The aim is to prove that $\boldsymbol{\beta_o}$ provides an improved objective value. As mentioned in the first step, the coefficients of objective function have positive second derivative. Hence, we obtain $d^{l_2}>d^{l_2-1}>d^{l_1+1}>d^{l_1}$ on condition that $l_2>l_2-1>l_1+1>l_1$. Again, we use $\Delta$ here to denote the change in objective function given by
		\begin{align} 
			\Delta&=-c^{l_1}\delta+c^{l_1+1}\delta+c^{l_2-1}\delta-c^{l_2}\delta \\
			&= \delta \left [{(c^{l_{1}+1}-c^{l_{1}}) - (c^{l_{2}} - c^{l_{2}-1}) }\right]  
			= \delta (d^{l_1}- d^{l_2-1}) < 0.
		\end{align}
		That ends the proof of a better feasible solution $\boldsymbol{\beta_o}$ than $\boldsymbol{\beta}$, which implies that the optimal solution has either only one non-zero variable, or two non-zero variables with consecutive indices. This remark largely simplifies the objective function and the constraints by reducing the number of variables to be optimized from $K+1$ to some around two. Recalling the conclusion in previous step, problem~\eqref{opt1} becomes
		\begin{subequations}\label{opt3} 
			\begin{align}  
				\min_{\{l,\beta^l, \beta^{l+1}\}} ~~~~& c^l\beta^l+c^{l+1}\beta^{l+1} \label{R3_0} \\
				{\rm s.t.}  ~~~~& \beta^l + \beta^{l+1} = 1, \label{R3_1}  \\
				~~~~& l \beta^l + (l+1) \beta^{l+1}=t, \label{R3_2} \\
				~~~~& 0 \le \beta^l \le \beta^{l+1} \le 1, \label{R3_3} \\
				~~~~& l \in [0:K-1]. \label{R3_4}
			\end{align}
		\end{subequations} 
		Problem \eqref{opt3} has three variables to be optimized with about four constraints two of which are equations. Jointly considering the constraints \eqref{R3_1}, \eqref{R3_3}, and \eqref{R3_4}, which imply non negative integers $l$ and $l+1$ as well as the proper fractions $\beta^l$ and $\beta^{l+1}$, constraint \eqref{R3_2} can thus be interpreted into 
		\begin{align}
			&l\beta^l + (l+1) \beta^{l+1}=t, \label{tandl0}\\
			&l \le l \beta^l + (l+1) \beta^{l+1} \le l+1, \label{tandl1}\\
			&l \le t \le l+1. \label{tandl2} 
		\end{align}
		According to \eqref{tandl2}, $t$ must falls within the interval spanned by $l$ and $l+1$. If $t \in \mathbb{Z}$ which is a common assumption for coded caching, one can easily derive that $t=l$ or $t=l+1$. However, no matter in which of the two cases mentioned, the optimal solution to problem \eqref{opt1} has only one non-zero variable at the index $l=t$ equivalently. For instance, when $t=l+1$, it follows that $\beta^{l+1}=\beta^{t}=1$. We then summarize the optimal solution as $\beta^t=1$, and $\beta^l=0, \text{when}~l \ne t$, which agrees with \eqref{sol2} in Theorem~\ref{theorem-opt}. If $t \notin \mathbb{Z}$, the relationship among $\{l, t, l+1\}$ is derived as $l=\floor*{t}$ and $l+1=\ceil*{t}$ which is the unique solution. Letting $\beta^{l}=\eta$, we obtain from \eqref{tandl0}-\eqref{tandl2} 
		\begin{align}
			t&=\floor*{t} \eta+\ceil*{t} (1-\eta) \notag  \\
			&=(\floor*{t}-\ceil*{t}) \eta +\ceil*{t}=-\eta +\ceil*{t}.  
		\end{align} 
		Hence, it holds true that $\eta=\ceil*{t}-t$. The optimal solution to problem \eqref{opt1} when $t \notin \mathbb{Z}$ is given by $\beta^{\floor*{t}}=\eta,\beta^{\ceil*{t}}=1-\eta$ with $\eta=\ceil*{t}-t$, and $\beta^l=0, l \in [0:K] \setminus \{\floor*{t}, \ceil*{t}\}$. That ends the proof of the optimal solution to the coded caching optimization problem \eqref{opt1}.
	    \end{proof}
       
    \begin{lemma}\label{lemma-step4}
    There is an exceptional case $I=K-1$ when the second derivative equals $0$. It will destroy the proof of Lemma~\ref{lemma-step3} which strictly requires a positive second derivative. The optimal solution, in this case, is no longer unique, but the solution in the general case still works. The same solution is claimed as the optimal one when $I=K-1$ for consistency.  
    \end{lemma}   
      
    \begin{proof}
		As mentioned previously, when $I=K-1$, the second derivative of the coefficients in the objective equals to $0$ which affects the proof of the structure of the optimal solution in Lemma~\ref{lemma-step3}. To this end, we firstly derive the objective function when $I=K-1$ as
		\begin{equation*} 
		\sum_{l=0}^{K-1} \frac{K-l}{l+1} \beta^l-\sum_{l=0}^{K-2} {K-1 \choose l+1}/{K \choose l}\beta^l=\sum_{l=0}^{K-1} \frac{(K-l)}{K} \beta^l. 
		\end{equation*}
	 Based on Lemma~\ref{lemma-step2}, the constraints \eqref{R1_1}-\eqref{R1_3} simplifies into 
	 \begin{align}\label{step4-cons}
	 \left\{\begin{array}{cl}
	 \sum_{l=0}^K \beta^l=1, &   
	  \sum_{l=1}^K l \beta^l = t,  \\
	 0 \le \beta^l,&l \in [0:K].  
	 \end{array}\right. 
     \end{align}

	Substituting \eqref{step4-cons} into the objective function, it becomes 
		\begin{equation}
		    \label{objK-1}
			\sum_{l=0}^{K-1} \frac{(K-l)}{K} \beta^l=1-\frac{t}{K}.
	\end{equation}
		Apparently, all feasible objective value is constant in this case. That is to say, the optimal objective value is $1-t/K$, and the optimal solution can be any $\boldsymbol{\beta}$ satisfying \eqref{step4-cons}. Without loss of generality, we state that the optimal solution when $I=K-1$ is the same as the one given in \eqref{sol2} and \eqref{sol3}. \end{proof}

	\begin{theorem}\label{theorem-opt} 
		The optimal solution to problem \eqref{opt1} is the same as the optimal solution to the file subpacketization optimization with fixed cardinality, which is given by  
		\begin{align} \label{sol2}
			\beta^l=\left\{\begin{array}{cl}
				1, &\mbox{if}~~ l=t, \\
				0, & \mbox{else},
			\end{array}\right. 
		\end{align}
		assuming $t=KM/N$ is integer. When $t=KM/N$ is non-integer, there are two adjacent nonzero elements in $\{\beta^l\}$ around $t$. Letting $\eta=\ceil*{t}-t$, the solution becomes  
		\begin{align} \label{sol3}
			 \beta^l=\left\{\begin{array}{cl}
				\eta, &\mbox{if}~~ l=\floor*{t}, \\
				1-\eta, &\mbox{if}~~ l=\ceil*{t},\\
				0, & \mbox{else}.
			\end{array}\right. 
		\end{align} 
	\end{theorem} 
	
	\begin{proof}
	This follows directly from following the four steps given in Lemma~\ref{lemma-step1}-Lemma~\ref{lemma-step4}.
   \end{proof}	
		
	   \begin{corollary}  
       Based on Theorem \ref{theorem-inactive-t} and Theorem \ref{theorem-opt},  Alg.~\ref{alg_cen} is optimal for integer $KM/N$. 
	   \end{corollary}	


	\section{Decentralized coded caching in presence of inactive users}\label{method4} 
	In practical wireless networks, one can rarely expect the server has totally central coordinating among the users, which has motivated the exploration on using a decentralized pattern for content placement, at the same time enjoying the global caching gain of coded multicasting in delivery. 
	
	Decentralized coded caching utilizes random content placement without coordination among users, so that it is easy to implement with acceptable performance degradation over the centralized scheme \cite{MaddahAli2014} ignoring user inactivity. Decentralized cache placement is applicable in our settings, where there is uncertainty from user inactivity, i.e. the number and identity of inactive users are unknown in the content placement phase. Compared to the centralized placement discussed previously, here each user fills local storage independently and randomly, so that the user inactivity information is unnecessary.  
	
	Decentralized coded caching proceeds in two phases: random content placement phase and content delivery phase. In placement phase, each user fills its cache with $\frac{MF}{N}$ bits of each file $n$ independently at random. Hence, no file splitting is needed. Here, we utilize notation ${\cal S} \subset [K]$ to denote a set of users with cardinality $|{\cal S}|=s$. $V_{k, {\cal S}}$ denotes the bits of the file requested by user $k$ (i.e. $d_k$) if the bits are cached exclusively by users in ${\cal S}$. That is to say, the bits in $V_{k, {\cal S}}$ are cached only by every user in ${\cal S}$, so that they are missing in users outside ${\cal S}$. Each bit is chosen at random uniformly with the probability denoted by $q \triangleq M/N \in (0,1)$. The probability for a bit is cached at any given user set ${\cal S}$ with cardinality $s$ becomes $Q^s=q^s(1-q)^{K-s}$. The number of bits being cached at the users in the given ${\cal S}$ is $F Q^s=F q^s(1-q)^{K-s}$. Accordingly, for any user $k \in {\cal S}$, the expected number of missing bits of file $d_k$ related to a given ${\cal S}$, i.e. $V_{k, {\cal S} \setminus \{k\} }$ is $F Q^{s-1}=F q^{s-1} (1-q)^{K-s+1}$. The benefit of focusing on set ${\cal S} \setminus \{k\}$ is that it enables us to use $V_{k, {\cal S} \setminus \{k\}}$ to identify the required bits for user $k$ corresponding to ${\cal S}$, and at the same time use $V_{j, {\cal S} \setminus \{j\}},~j \in {\cal S}, j \ne k$ to describe the cached bits at user $k$. For any possible set ${\cal S}$ with cardinality $s$, the server needs to send the maximum missing bits for the users in ${\cal S}$ given by
	\begin{equation}
		\max_{k \in {\cal S}} |V_{k, {\cal S} \setminus \{k\} }|=F q^{s-1} (1-q)^{K-s+1}.
	\end{equation}
	In total, there are ${K \choose s}$ types of distinct sets ${\cal S}$ with given cardinality $s$. In delivery, one should sum up all the distinct  ${\cal S}$ with given cardinality $s$ in a loop, and then let cardinality $s$ vary from $1$ to $K$. 
	
	Without user inactivity, the transmitted packet with user set ${\cal S}$ is $\oplus_{k \in {\cal S}}~V_{k, {\cal S} \setminus \{k\}}$. The backhaul load is obtained as 
	\begin{equation}\label{Rdecen}
		R=\sum_{s=1}^{K} {K \choose s} q^{s-1} (1-q)^{K-s+1} 
		=\frac{1-q}{q}\big(1\!-\!(1-q)^K\big).
	\end{equation}
	
	Considering user inactivity, the server determines the bits to be transmitted utilizing the information of user inactivity. Similar to 
	(\ref{R_inact-obj}), only the active users require to be served. Consequently, there are two types of ${\cal S}$ in terms of the cardinality $s$. First, when the cardinality $s \in [I+1:K]$ which guarantees at least one active user in ${\cal S}$, the transmitted bits with regarding to ${\cal S}$ are $\oplus_{k \in {\cal S}, k \notin {\cal I}}~V_{k, {\cal S} \setminus \{k\}}$. In contrast, when $s \in [1:I]$, it may happens that all the users in ${\cal S}$ are inactive so that there is no need to consider set ${\cal S}$. The transmitted bits can be written as $\oplus_{k \in {\cal S}, k \notin {\cal I}}~V_{k, {\cal S} \setminus \{k\}, {\cal S} \not\subset {\cal I} }$. The procedure of the decentralized coded caching in presence of user inactivity is given in Alg.~\ref{alg_dec} with delivery policy I. 
	
	\begin{algorithm} 
		\caption{Decentralized Coded Caching in Presence of User Inactivity}
		\label{alg_dec}  
		\begin{algorithmic}[1]
			\STATE \textbf{procedure} {\large P}LACEMENT
			\STATE ~~~~\textbf{for} $k \in [K], n \in [N]$ do\\
			\STATE~~~~~~~user $k$ independently caches $\frac{MF}{N}$ bits of file $n$, \\
			~~~~~~~chosen uniformly at random\\
			\STATE~~~~\textbf{end for}
			\STATE \textbf{end procedure} \\
			Users make requests $\boldsymbol{d}$ given the number and identity of the inactive users $(I, {\cal I})$; \\
			The number of active users $J=K-I$ 
			\STATE \textbf{procedure} {\large D}ELIVERY I
			\STATE ~~~~\textbf{for} $s=K, K-1,\dots, I+1$ do\\
			\STATE~~~~~~~\textbf{for} ${\cal S} \subset [K]$ with $|{\cal S}|=s$ do\\
			\STATE~~~~~~~~~~~server sends $\oplus_{k \in {\cal S}, k \notin {\cal I}}~V_{k, {\cal S} \setminus \{k\}}$
			\STATE~~~~~~~\textbf{end for}
			\STATE ~~~~\textbf{end for}
			\STATE ~~~~\textbf{for} $s=I, I-1,\dots, 1$ do\\
			\STATE~~~~~~~\textbf{for} ${\cal S} \subset [K]$ with $|{\cal S}|=s$ do\\
			\STATE~~~~~~~~~~~server sends $\oplus_{k \in {\cal S}, k \notin {\cal I}}~V_{k, {\cal S} \setminus \{k\}, {\cal S} \not\subset {\cal I}}$		
			\STATE~~~~~~~\textbf{end for}
			\STATE ~~~~\textbf{end for}
			\STATE \textbf{end procedure}
			\STATE \textbf{procedure} {\large D}ELIVERY II
			\STATE ~~~~\textbf{for} $s=J, J-1,\dots, 1$ do\\
			\STATE~~~~~~~\textbf{for} ${\cal S} \subset [K]\setminus {\cal I}$ with $|{\cal S}|=s$ do\\
			\STATE~~~~~~~~~~~server sends $\oplus_{k \in {\cal S}}~V_{k, {\cal S} \setminus \{k\}}$
			\STATE~~~~~~~\textbf{end for}
			\STATE ~~~~\textbf{end for}
			\STATE \textbf{end procedure}
		\end{algorithmic}
	\end{algorithm}
	
	The backhaul load considering user inactivity becomes
	\begin{align}\label{RdecenInact}
		R&=\left[\sum_{s=I+1}^{K} {K \choose s} +\sum_{s=1}^{I} \left[{K \choose s}-{I \choose s}\right] \right]  
		q^{s-1} (1-q)^{K-s+1}  \notag \\
		&=\frac{1-q}{q}\big(1-(1-q)^{K-I}\big). 
	\end{align}
	The connection between the backhaul loads with user inactivity $(I, \cal{I})$ and without user inactivity can be derived as 
	\begin{equation}\label{RdecenConnect}
		R(I, K)=R(0,K-I)<R(0,K). 
	\end{equation}
	Here $R(a, b)$ is the load with $a$ inactive users and $b$ total users.  
	
	Equation \eqref{RdecenConnect} implies that the scenario with $I$ inactive users out of in total $K$ users equals the case when there are original $J=K-I$ users and all are active. The remark agrees with the key idea of decentralized caching where the content placement operates at each user separately without coordination. Consequently, a more direct way is to replace $(K, [K])$ with the number and identities of the active users $(J, [K] \setminus {\cal I})$ when defining user set ${\cal S}$ in delivery. In this case, there is no need to worry about inactive users. The alternative delivery policy is given in Delivery II in Alg.~\ref{alg_dec}. 
	
	\section{Comparison of centralized and decentralized cache placement against user inactivity}
	\subsection{Centralized versus Decentralized Coded Caching}
	Here we briefly compare the backhaul loads derived using decentralized method in Alg.~\ref{alg_dec} and the centralized method in Alg.~\ref{alg_cen}, in terms of the number of inactive users $I$. 
	
	First, it can be concluded from the key ideas and precise backhaul loads of the centralized  \eqref{R_inact-obj} and decentralized \eqref{RdecenInact} versions that:
	\begin{itemize}
		\item The backhaul loads in both manners decrease when the number of inactive users $I$ rises. 
		\item When $I \to 0$, the backhaul loads for both manners approximate the backhaul loads without user inactivity.
		\item On the contrary, when $I \to K$, there are few users staying active, so that both of the backhaul loads approach $0$.
		\item The gain of centralized manner over decentralized manner results from the coordination among active users which increases the opportunities for coded multicasting, so that it decreases with regarding to $I$. Particularly when there is only one active user, the two manners become the same. 
	\end{itemize}
	
	While the qualitative discussion is presented above, the gain of centralized coded caching over decentralized manner against user inactivity defined as $G=R_d-R_c$ is investigated quantitatively. Here index $c$ and $d$ denote centralized, and decentralized manners respectively. We obtain  
	\begin{align*} 
		G(I)=\left\{\begin{array}{cl}
		 \frac{N-M}{M}\big(1-(1-\frac{M}{N})^{K-I}\big)-\frac{K(1-M/N)}{1+KM/N}, &  \mbox{if}~t+1!> I,\\
	 	 \frac{N-M}{M}\big(1-(1-\frac{M}{N})^{K-I}\big)-\frac{1}{{K \choose t}} \left[{K \choose t+1}-{I \choose t+1}\right], & \mbox{if}~t+1\le I. 
		\end{array}\right. 
	\end{align*} 
	It holds true that $G(K)=R_d(K)=R_c(K)=0$. Now we prove that $G(I)$ decreases with regarding to $I$, i.e. $G(I) \ge G(I+1)$, such that the maximum gain is given by $\max_I G(I)=G(1)$. The first derivative of $G$ is  
	\begin{equation}\label{Gcd}
		\Delta_G(I)=G(I+1)-G(I)=\Delta_{R_d}(I)-\Delta_{R_c}(I),
	\end{equation}
	where $\Delta_{R_d}(I)$ and $\Delta_{R_c}(I)$ are defined similar to $\Delta_G(I)$.
	
	As $R_c$ in \eqref{R_inact-obj} is piece-wise, the proof is operated for the two intervals separately. 
	
	When $I \in [t+1:K]$, We  derive the differences as 
	\begin{align*}
		\Delta_{R_d}(I)&=-\left(1-\frac{M}{N}\right)^{K-I}, \\
		\Delta_{R_c}(I)&=\frac{1}{{K \choose t}} \left[{I \choose t+1}-{I+1 \choose t+1} \right] 
		\overset{(a)}{=}\!-{I \choose t}/{K \choose t}\\
		&\overset{(b)}{\ge}-\left(\frac{K-t}{K}\right)^{K-I} 
		\ge \Delta_{R_d}(I),
	\end{align*}
	where $(a)$ is based on equation ${n+1 \choose k}={n \choose k}+{n \choose k-1}$ with $n, k$ are both positive integers, while $(b)$ is obtained as $\frac{K-t}{K}$ decreases in terms of $K$. In this case, it has been verified that $\Delta_G(I) \le 0$ when $I \in [t+1:K]$. In particular, the equality is achieved when $I=K-1$.  
	
	Moreover, when $I \in [1:t+1)$, $R_c=\frac{K(1-M/N)}{1+KM/N}$ is independent on the number of inactive users $I$, and hence $R_c$ remains constant which means that $\Delta_{R_c}(I)=0$ in the considered interval. Then it is apparent that $\Delta_G(I)=\Delta_{R_d}(I) < 0$.
	
	Finally, we verify that $\Delta_G(I) \le 0$ for all possible $I$. The maximum gain of centralized versus decentralized coded caching against user inactivity is given by $\max_I G=G(1)$.
	
	Consequently, we obtain that $0 \le G(I+1) \le G(I) \le G(1)$ for the consider interval $I \in [K-1]$, which can be observed in the simulation results in Subsection~\ref{sim3}.
	
	\subsection{Proposed Opt CC versus Ideal MAN Method}
	
	In this section, we aim to provide some insights on how the proposed centralized method, referred to as \textit{optimization based coded caching (Opt CC)}, performs compared with the ideal MAN method which assumes perfect user inactivity information in the content placement phase. 
	
	The gap between the backhaul load of Opt CC using \eqref{R_inact-obj} and the one derived using the ideal MAN method according to \eqref{R} is discussed here, which reflects the performance loss of Opt CC for lack of user inactivity information in the placement phase. We set $\tilde{G}=R_c-R_i$ with the backhaul load of the ideal MAN method given by 
	\begin{equation}\label{R_ideal}	
		R_i=\frac{1}{{J \choose \tilde{t}}} {J \choose \tilde{t}+1}
		=\frac{(K-I)\left(1-M/N\right)}{1+(K-I)M/N}, 
	\end{equation}
	where the total number of users $K$ is replaced by the number of active users $J=K-I$ while $t$ is updated to $\tilde{t}=JM/N$. 
	
	Next, we analyze how $\tilde{G}$ varies with regarding to $I$. $\tilde{G}$ is  
	\begin{align}\label{G1}
		\tilde{G}=\left\{\begin{array}{cl}
		 \frac{K(1-M/N)}{1+KM/N}-\frac{(K-I)\left(1-M/N\right)}{1+(K-I)M/N}, &\mbox{if}~t+1> I, \\
			\frac{1}{{K \choose t}} \left[{K \choose t+1}-{I \choose t+1}\right]-\frac{(K-I)\left(1-M/N\right)}{1+(K-I)M/N}, &  \mbox{if}~t+1 \le I. 
		\end{array}\right. 
	\end{align} 
	\subsubsection{Performance gap analysis against $I$ when $I \in [1:t+1)$}  \label{subsec-gap1}
	\hfill
	
	When $I \in [1:t+1)$, it is apparent that $\tilde{G}$ increases with regarding to $I$. When $I \to 0$, it holds true that $\lim_{I \to 0} \tilde{G}=0$ which shows that the Opt CC approximates the ideal MAN method with very low user inactivity. In general, we derive that $\tilde{G}(I+1) > \tilde{G}(I) >0$ for $I \in [1:t+1)$, and also $\lim_{I \to 0} \tilde{G}=0$. The observation agrees with the fact that the performance loss of Opt CC results from the lack of user inactivity and therefore increases with regarding to $I$. In this case, the maximum gap is $\tilde{G}(t)$ given by 
	\begin{equation}
		\tilde{G}(t)=\frac{1}{1+t}\left(1-\frac{1}{1+t-tM/N}\right).
	\end{equation}
	
	\subsubsection{Performance gap analysis against $I$ when $I \in [t+1:K)$}   \label{subsec-gap2}   
	\hfill
	
	We firstly consider the extreme case that 
	$\lim_{I \to K} \tilde{G}=0$. Moreover, we simplify the gain at starting point $I=t+1$ as 
	\begin{align}\label{tildeGt1}
		\tilde{G}(t+1)
		= &\frac{1-M/N}{1-M/N+(K-t)M/N}-\frac{1}{{K \choose t}} \notag \\
		= &\frac{1}{t+1}-\frac{(K-t)(K-t-1)\dots 2}{K(K-1)\dots(t+2)} \times \frac{1}{t+1} > 0, 
	\end{align}
	 To investigate the monotonicity of $\tilde{G}$ in this interval, we obtain the difference  $\Delta_{\tilde{G}}=\tilde{G}(I+1)-\tilde{G}(I)$ as 
	\begin{equation*} 
		\Delta_{\tilde{G}}=-\underbrace{\frac{{I \choose t}}{{K \choose t}}}_{\Delta_{\tilde{G}_1}}+\underbrace{\frac{\left(1-M/N\right)}{\left(1+(K-I)M/N\right)\left(1+(K-I-1)M/N\right)}}_{\Delta_{\tilde{G}_2}}.
	\end{equation*} 
	The first derivatives are obtained as 
	\begin{align*}
		&\Delta_{\tilde{G}_1}'(I)=\Delta_{\tilde{G}_1}(I+1)-\Delta_{\tilde{G}_1}(I)=\frac{{I \choose t}}{{K \choose t}} \frac{t}{I+1-t} >0,  \\
		&\Delta_{\tilde{G}_2}'(I)=\Delta_{\tilde{G}_2}(I+1)-\Delta_{\tilde{G}_2}(I)~~~~~~~~~~~~~~~~~~~~~~~~~~\\
		&=\frac{(1-M/N)(2M/N)}{(1+(K-I-2)M/N)(1+(K-I-1) M/N)(1+(K-I)M/N)}, 
	\end{align*}	
	which states both $\Delta_{\tilde{G}_1}$ and $\Delta_{\tilde{G}_2}$ increase with regarding to $I$. \footnote{Notation $(*)$ in $\Delta_{\tilde{G}_i}(*)$ specifies the values of $\Delta_{\tilde{G}_i}$ when $I=*$.} 
	Moreover, the corresponding second differences satisfy $\Delta_{\tilde{G}_1}''>0$,  $\Delta_{\tilde{G}_2}''>0$ because both $\Delta_{\tilde{G}_1}'$ and $\Delta_{\tilde{G}_2}'$ monotonically increases with regarding to $I$. 
	
	Now we discuss the ending and starting points of $\Delta_{\tilde{G}}$. The details of the ending point at $I=K-1$ can be obtained as
	\begin{align*}
		\Delta_{\tilde{G}_1}(K-1)&=1-M/N, \\
		\Delta_{\tilde{G}_2}(K-1)&=\frac{1-M/N}{1+M/N} < \Delta_{\tilde{G}_1}(K-1), 
	\end{align*}
	with the corresponding first derivative given by 
	\begin{align*}
		\Delta_{\tilde{G}_1}'(K-1)&=\frac{{K-1 \choose t}}{{K \choose t}} \frac{t}{K-t}=\frac{M}{N},~~~\\
		\Delta_{\tilde{G}_2}'(K-1)
		&=\frac{M}{(N+M)/2} >\Delta_{\tilde{G}_1}'(K-1).~~
	\end{align*}
	Similarly, we discuss the starting point when $I=t+1$ with the conclusion presented in Lemma~\ref{deltaGstart}. Note that the sign of $\Delta_{\tilde{G}}(t+1)$ is not fixed. It depends on the values of $t$ and $K$. 
	\begin{lemma}\label{deltaGstart}
		Assuming the starting point defined as $\Delta_{\tilde{G}}(t+1)=\Delta_{\tilde{G}_2}(t+1)-\Delta_{\tilde{G}_1}(t+1)$ with $t \in [1:K-2]$ when $I=t+1$ , it holds true that
		\begin{align}\label{deltaGtplus1} 
			\Delta_{\tilde{G}}(t+1) \left\{
			\begin{array}{cl}
				> 0, & \mbox{if}~~t \in [1:K-3], K \in [9:+\infty],\\
				> 0, & \mbox{if}~~t \in [1:K-4], K \in [6:8],\\
				< 0, & \mbox{if}~~t=K-3,~K \in [6:8], \\
			    <0, & \mbox{if}~~t \in [1:K-3], K \in [4:5].\\
				<0, & \mbox{if}~~t=K-2, K \in [3:+\infty]. 
			\end{array} \right. 
		\end{align}   
	\end{lemma}
	
	\begin{proof}
		See Appendix D.
	\end{proof}

	It is now proved that both $\Delta_{\tilde{G}_1}$ and $\Delta_{\tilde{G}_2}$ increase in terms of $I$ with increasing growth rates. According to the sign of the starting point $\Delta_{\tilde{G}}(t+1)$, there are two cases regarding the impact of $I$ on $\Delta_{\tilde{G}}$ based on the derivatives:  
	\begin{itemize} 
		\item Case I: When $\Delta_{\tilde{G}_2}(t+1) > \Delta_{\tilde{G}_1}(t+1)$, there is an intersection point for $\Delta_{\tilde{G}_1}$ and $\Delta_{\tilde{G}_2}$ in the interval. Hence, $\tilde{G}$ goes up from $I=t+1$ until reaching a peak, and then falls to $0$ at $I=K$. Only one peak point is guaranteed at some $I^*$. Meanwhile, $\Delta_{\tilde{G}} \ge 0$ during $I \in [t+1:I^*]$, and $\Delta_{\tilde{G}} <0$ when $I \in (I^*:K)$.  
		\item Case II: When $\Delta_{\tilde{G}_2}(t+1) \le \Delta_{\tilde{G}_1}(t+1)$, it holds true that $\Delta_{\tilde{G}}(t+1)<0$ in the interval. Considering the ending point and derivatives, it is derived that $\Delta_{\tilde{G}} <0$ during the interval $[t+1:K-1]$. It means that $\tilde{G}$ decreases with regarding to $I$ during the interval. The maximum value of $\tilde{G}$ in the  interval is $\tilde{G}(t+1)$ given in \eqref{tildeGt1}. 
	\end{itemize}
	
	\subsubsection{Summary of performance gap $\tilde{G}$ against $I$}\hfill
	
	Now we summarize the performance gap between $R_c$ and $R_i$ against $I$ when $I \in [K-1]$ in Theorem~\ref{theorem-gap}. 
	
   \begin{theorem}\label{theorem-gap}
	The value of $\tilde{G}$ defined in \eqref{G1} varies with different values of $I$ in the following manner:
	\begin{itemize} 
		\item If $I \in [1:t+1)$, $\tilde{G}$ increases in terms of $I$ as mentioned previously. In particular, $\lim_{I \to 0} \tilde{G}=0$ which agrees with the fact that Opt CC performs the same as the ideal MAN method without user inactivity. 
		\item If $I \in [t+1:K)$, the starting point satisfies $\tilde{G}(t+1) >0$ and the ending point follows $\lim_{I \to K} \tilde{G}=0$. 
		\item Meanwhile, $\tilde{G}$ may continue rising from $(I=t+1)$ until reaching the peak at $(I=I^*)$, and then decreases to $0$ at $(I=K-1)$. Only one peak point is guaranteed at $I^*$ during the whole interval. Alternatively, it probably starts to decrease with regarding to $I$ from $(I=t+1)$ with the maximum value being $\tilde{G}(t+1)$.
	\end{itemize}
	\end{theorem}
	\begin{proof}
     To prove Theorem~\ref{theorem-gap}, follow the discussion in subsection~\ref{subsec-gap1} and subsection~\ref{subsec-gap2}.  
	\end{proof}
	
	\section{Simulations}\label{sim}  
	\begin{figure}[htbp] 
		\centering
		\includegraphics[width=13cm]{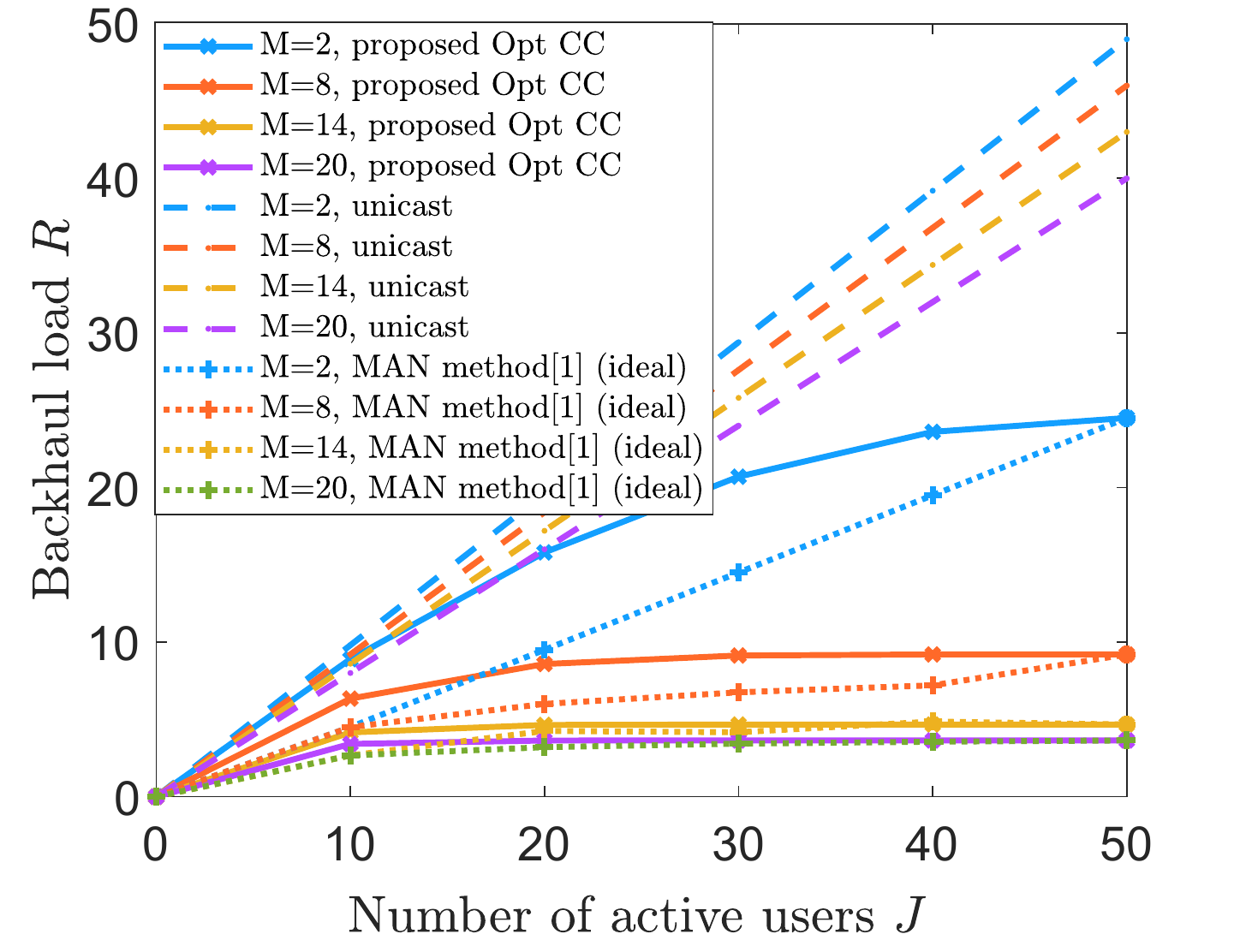}
		\caption{Backhaul load of coded caching network with user inactivity}\label{figsim_1}  
	\end{figure}
	\subsection{The Proposed Coded Caching Method with fixed cadinality}\label{sim1}
	A one-server cache-enable network is considered. There are $K=50$ users and $N=100$ files with equal popularity and size. Fig.~\ref{figsim_1} presents the impacts of cache size $M$ and number of active users $J$ on worst case backhaul load. The performance of the proposed optimization based coded caching with fixed cardinality is compared to the unicast caching method, and the ideal MAN method~ \cite{M.MaddahAliMay2014} where perfect user inactivity information is assumed in the cache placement phase.   
	
	As can be seen Fig.~\ref{figsim_1}, the proposed scheme outperforms the unicast caching scheme while provides a compatible performance to the idea MAN method against user inactivity. In general, the backhaul load decreases with the increase of cache size $M$ and rises with the increase of the number of active users $J$. Increasing $M$, the backhaul load decreases dramatically. The decrease of backhaul load is more obvious with more active users, i.e. larger $J$. The gap between them decreases when the cache size $M$ rises. When $M=20$, the proposed method is approximately the same as the MAN method in the ideal case. Increasing $J$ causes some increase in backhaul load, but the gap is limited for the proposed method, e.g. less than $10$. The gap is getting narrower when increasing cache size $M$. That is to say, the proposed method can provide an acceptable solution to user inactivity. Particularly when $J=K=50$, the proposed method is the same as the MAN method as all the users are active. Moreover, the backhaul load tends to be stable when $J$ tends to $K$, which agrees with \eqref{R_inact-obj} since the backhaul load is independent of the number of inactive users $I$ when $I<t+1$, i.e. $J>K-t+1$. For instance, when $M=8$, the solid curve in orange becomes stable after reaching $J=43$. 
	\subsection{Optimal Coded Caching Method with multiple cardinalities}\label{sim2}
	To investigate the optimal coded caching method in presence of user inactivity, the performance of coded caching using fixed cardinality and multiple cardinalities in file subpacketization is compared. Results obtained using optimization solver CVX are compared to the closed form optimal solution~\eqref{sol2}-\eqref{sol3}. 
	
	\begin{figure}[htb] 
		\centering
		\includegraphics[width=13cm]{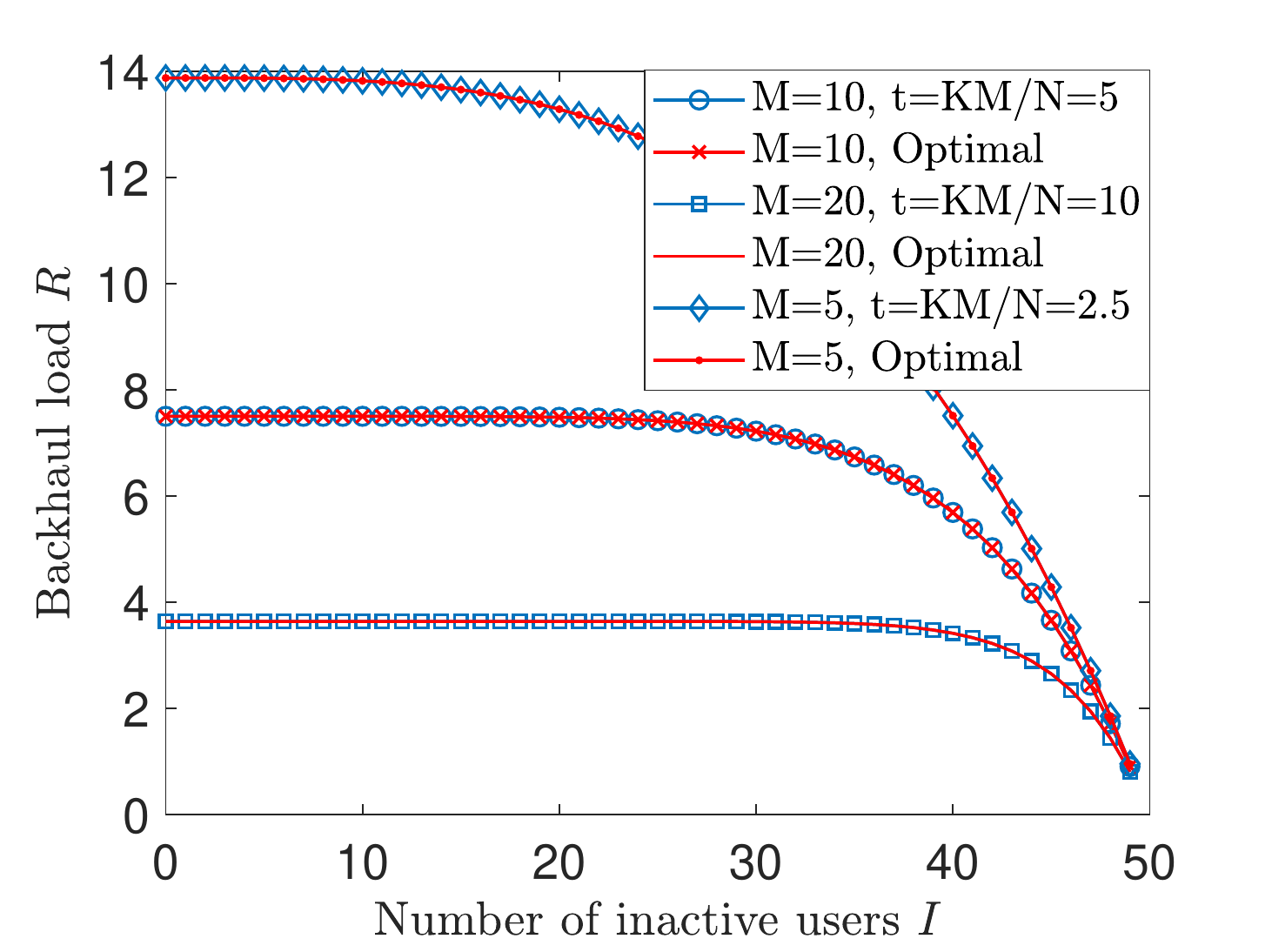}
		\caption{Performance comparison between coded caching with fixed $l$ and optimal scheme with multiple cardinalities against user inactivity.}\label{figsim_2}  
	\end{figure} 
	
	In Fig.~\ref{figsim_2} we consider a cache enabled network where we let $K=50, N=100, M=5/10/20$, and the number of inactive users varies within $[0:K)$, to compare the performance of the optimal solution for coded caching with multiple $l \in [0:K]$ and the one with fixed $l=t$. 
	The simulation confirms the closed form solution. The backhaul load for the proposed caching scheme with fixed cardinality is always the same as the optimal solution with multiple $l$ for all the different values of inactive users and cache size.
	
	\subsection{Centralized Method versus Decentralized Method}\label{sim3}
		\begin{figure}[htb]   
		\centering
		\includegraphics[width=13cm]{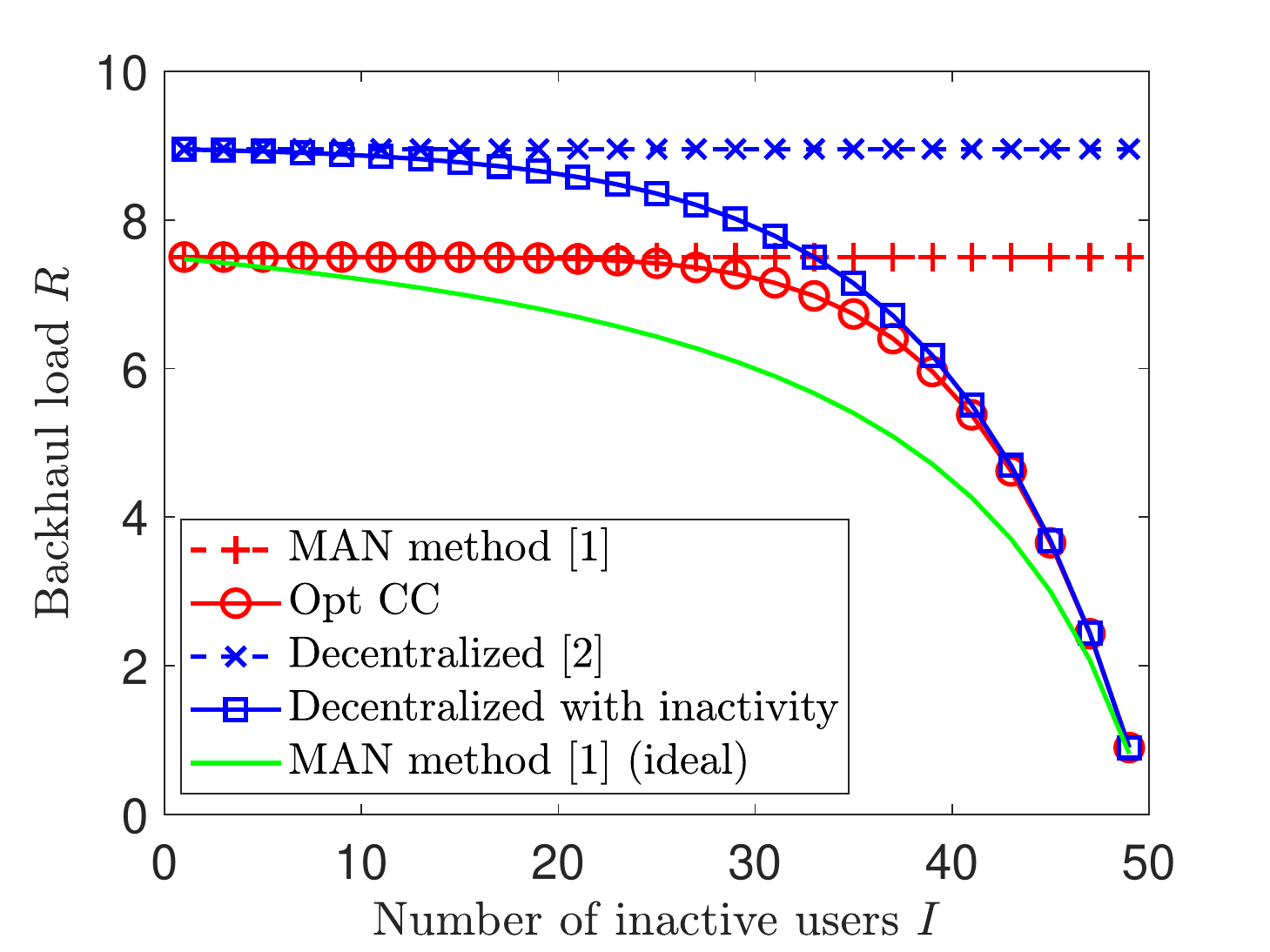}
		\caption{Performance comparison between coded caching methods with different numbers of inactive users $I$: centralized versus decentralized ($M=10, K=50$).}\label{figsim_3}  
	\end{figure} 
	
	\begin{figure}  
		\centering
		\includegraphics[width=13cm]{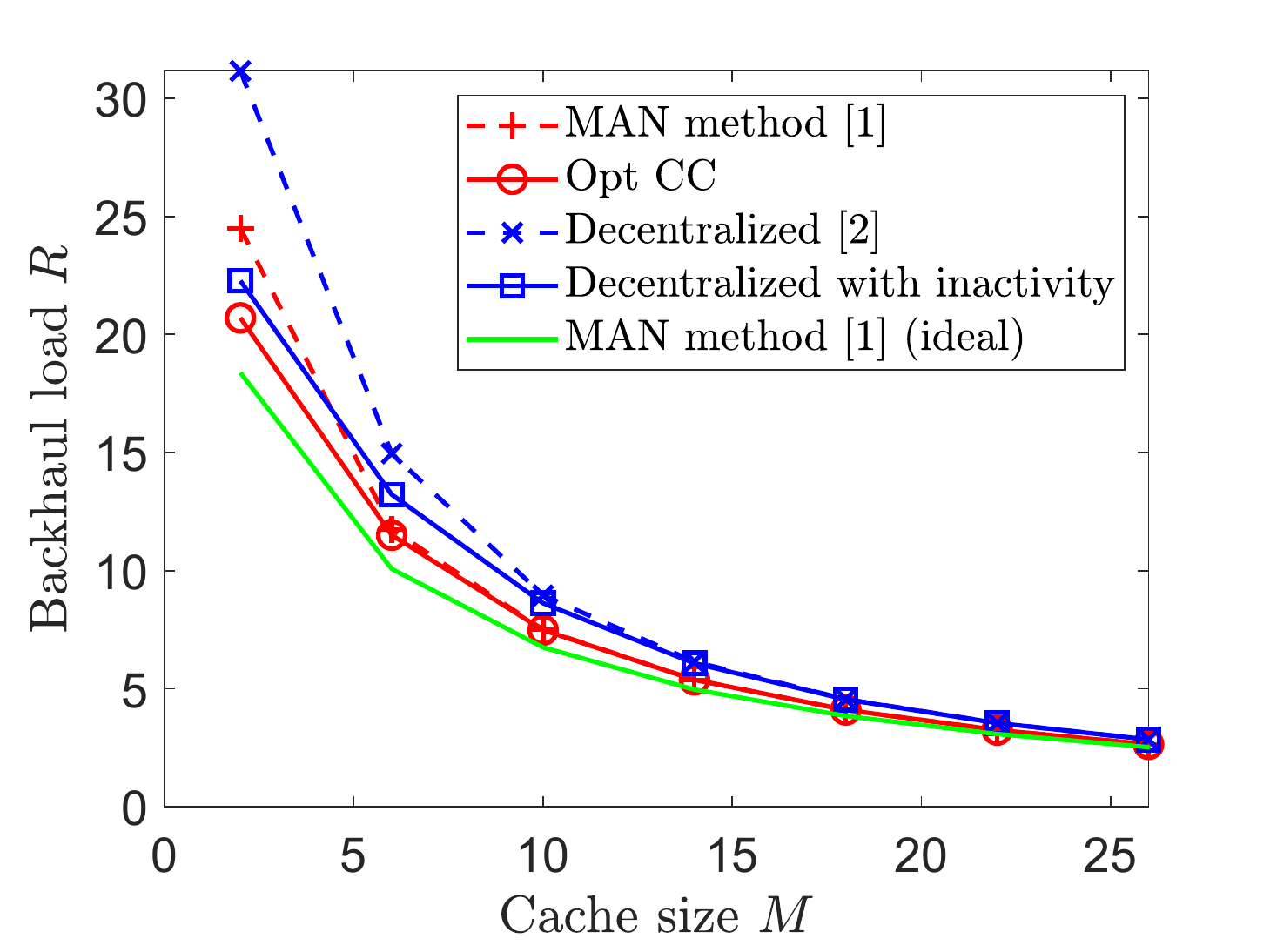}
		\caption{Performance comparison between coded caching methods with different normalized cache sizes $M$: centralized versus decentralized ($I=20, K=50$).}\label{figsim_4}  
	\end{figure} 
	
	In this section, we compare the performance of centralized and decentralized coded caching methods. Performance is measured in terms of worst case backhaul load $R$ with the impacts of the number of inactive users $I$ and cache size $M$ investigated. In Fig.~\ref{figsim_3} and Fig.~\ref{figsim_4}, we have $K=50, N=100$. 
	In Fig.~\ref{figsim_3} we fix $M=10$ and vary the number of inactive users $I\in [K-1]$, while in Fig.~\ref{figsim_4} we fix the number of inactive users at $I=20$ and consider different cache sizes $M \in [N/4]$.

   Five methods are compared. The centralized methods consist of: the MAN method without user inactivity, assuming $J=K=50$, with backhaul load calculated using \eqref{R}; the proposed optimization based coded caching (Opt CC) given in Alg.~\ref{alg_cen}, using \eqref{R_inact-obj}; as well as the ideal MAN method in \eqref{R_ideal} assuming perfect user inactivity information in placement phase. The decentralized methods consist of: a decentralized caching without user inactivity $J=K=50$ based on \eqref{Rdecen} and decentralized caching with user inactivity following Alg.~\ref{alg_dec} using \eqref{RdecenInact}. 
	
	As can be seen in Fig.~\ref{figsim_3}, in this scenario, centralized coded caching provides at most around $17\%$ gain over decentralized methods. Considering the impacts of the number of inactive users $I$, the proposed Opt CC outperforms the decentralized manner with user inactivity, and it is closer to the ideal centralized coded caching method. Moreover, the gain of Opt CC over the decentralized manner with user inactivity decreases from around $17\%$ to $0\%$ when the number of inactive users $I$ rises from $I=1$ to $K-1$. 
	The advantage of the centralized method results from coordination among users, which decreases when there are fewer active users to be served. 
	When approaching the minimum $J=1$ active user, the performance of the centralized and decentralized method gradually converge.  The gap between Opt CC and the ideal MAN method in Fig.~\ref{figsim_3} agrees with the first case in Section \ref{subsec-gap2}; the gap $\tilde{G}$ grows with increasing $I$ 
	until peaking at $I^* \approx 35$, and then decreases to $0$ at $I \to K$.

	Fig.~\ref{figsim_4} shows the dependence of backhaul performance on 
	the cache size $M$, with $I=20$. For all schemes of interest, the load decreases with increasing storage space as expected. Moreover, the rate of decrease shrinks with increasing cache size. Centralized Opt CC outperforms the decentralized methods with or without user inactivity, and the conventional MAN method while approximating the ideal MAN method.  
	
	\section{Conclusions}
	In this paper, we studied the coded caching strategy for one-server cache-enable networks in presence of inactive users. Based on the classic file subpacketization and coded multicast strategy, a coded caching method with fixed cardinality of the fragment label set has been proposed and the optimality of selected cardinality $t=KM/N$, which is known as the cache replication parameter used in MAN method without user inactivity, has also been proved. The scheme is extended to a scenario where multiple cardinalities can be used instead of a fixed one. The weights for different types of fragments labeled with different cardinalities are optimized in order to minimize the worst case backhaul load. The optimal solution turns out to be the same as the one derived from the proposed caching scheme with fixed cardinality. Besides centralized coded caching, a decentralized method is also discussed. Next, we compare the backhaul load as a function of the number of inactive users $I$, provided by the identified optimal centralized method, the decentralized method, and the ideal MAN method assuming perfect knowledge of user inactivity. Numerical results show that the optimization based centralized scheme outperforms the decentralized scheme in presence of user inactivity, and approximates the ideal MAN method at the same time. 
	
	\section*{Appendix}
	
	\subsection{Proof of Lemma~\ref{lemma-optFix-obj}}
	To proceed, we compute the ratio $R(l+1)/R(l)$, with $1\leq l \leq I-2$. By expanding the binomial coefficients we get
	\begin{eqnarray}
		\frac{R(l+1)}{R(l)} &=& \frac{\big[{K \choose l+2}-{I \choose l+2}\big]}{\big[{K \choose l+1}-{I \choose l+1}\big]} \frac{{K \choose l}} {{K \choose l+1}}  \cr
		&=&  \frac{  A(l) }{ A(l) + B(l)-C(l)} \nonumber
	\end{eqnarray}
	where 
	\begin{eqnarray}
		A  &=& (l+1) \left(\frac{K!}{(K-l-2)!}  - \frac{I!}{(I-l-2)!}\right)\cr
		B  &=& \frac{(K+1)!}{(K-l-1)!} \cr
		C &=& (K+1 +(K-I)(l+1))\frac{I!}{(I-l-1)!}\,.\nonumber
	\end{eqnarray}
	Because $K \ge I$, we have $A \ge 0$, as well as $B>0$ and $C>0$ in the domain of interest. It is thus sufficient to  prove  that $B>C$ for all $l$ in the domain. For this, we first use the fact that $(K-m)/(I-m) \geq K/I$ for non-negative $m$ and $K\geq I$ to lower bound the $l+1$ smallest terms in the $l+2$-fold product in $B$ to get
	\begin{equation}
	    B \geq (K+1)\left(\frac{K}{I}\right)^{l+1} \frac{I!}{(I-l-1)!}
	\end{equation}
	Writing $K/I = 1 + (K-I)/I$ we can expand $(K/I)^{l+1}$ using the binomial expansion. When $K>I$, all terms in the expansion are positive. Comparing the two first terms in this expansion to $C$ directly shows that $B>C$ holds. This completes the proof. \hfill  

	\subsection{Proof of Lemma~\ref{lemma-step1}}\label{proof-step1}	
	The objective function \eqref{R1_0} is reformulated to 
	\begin{equation}
		f({\boldsymbol \beta})=\sum_{l=0}^{K-1} c^l \beta^l=\sum_{l=0}^{I-1} c_1^l \beta^l + \sum_{l=I}^{K-1} c_2^l \beta^l, 
	\end{equation}
	where it is defined that ${\boldsymbol c} \triangleq [c^0,c^2,\dots,c^{K-1}]$ denoting the coefficients in the objective given by 
	\begin{align} \label{c}
		c^l=\left\{\begin{array}{cl}
			c_1^l=\frac{K-l}{l+1}-{I \choose l+1}/{K \choose l}, &  \mbox{if}~ 0 \le l \le I-1,   \\ 
			c_2^l=\frac{K-l}{l+1}, &  \mbox{if}~ I \le l \le K-1.  
		\end{array}\right. 
	\end{align} 
	According to \eqref{c}, the coefficient ${\boldsymbol c}$ is a piece-wise function in terms of $l$ consisting of two sub-functions, ${\boldsymbol c_1}$ for interval $l \in [0:I-1]$ and ${\boldsymbol c_2}$ for $l \in [I:K-1]$, respectively. We will start with the discussion of ${\boldsymbol c_2}$, then move to ${\boldsymbol c_1}$. 
	\subsubsection{The first derivative of coefficients ${\boldsymbol c}$} 
	It is apparent that ${\boldsymbol c_2}$ is monotonically decreasing in terms of $l$ in the interval. The first derivative ${\boldsymbol d_2}=[d_2^I,d_2^2,\dots, d_2^{K-1}]$ is 
	\begin{equation}
		d_2^l=c_2^{l+1}-c_2^l=\frac{-K-1}{(l+2)(l+1)}<0,\label{d2}
	\end{equation} 
	and the second derivative ${\boldsymbol e_2}=[e_2^I,e_2^2,\dots, e_2^{K-1}]$ is given by 
	\begin{equation}
		e_2^l=d_2^{l+1}-d_2^l=\frac{2(K+1)}{(l+3)(l+2)(l+1)}>0.\label{e2}
	\end{equation} 
	Thus, it implies a negative first derivative and a positive second derivative for ${\boldsymbol c_2}$.  
	
	Now we discuss the monotonicity of ${\boldsymbol c_1}$. We rewrite ${\boldsymbol c_1}$ into 
	\begin{equation}
	  	c_1^l=\bigg[{K \choose l+1}-{I \choose l+1}\bigg]/{K \choose l},  \label{c1}  	\end{equation}
	utilizing $\frac{K-l}{l+1}={K \choose   l+1}/{K \choose l}$. Recalling the result from Appendix B that  
	\begin{equation}
		\frac{R(l+1)}{R(l)}=\frac{\big[{K \choose l+2}-{I \choose l+2}\big]}{\big[{K \choose l+1}-{I \choose l+1}\big]} \frac{{K \choose l}} {{K \choose l+1}} <1,  \label{B_R}
	\end{equation}
	it follows that $0 < c_1^{l+1}/c_1^l <1$. Hence, the first derivative of ${\boldsymbol c_1}$, referring to as ${\boldsymbol d_1}$ satisfying  
	$d_1^l=c_1^{l+1}-c_1^l < 0$.
	
	Since we have confirmed that ${\boldsymbol c_1}$ and ${\boldsymbol c_2}$ are monotonically decreasing in their own intervals, now we need to prove $c_1^{I-1} > c_2^I$ holds true. To proceed, we simplify $c_1^{I-1}$ to   
	\begin{equation}\label{breakp0}
		c_1^{I-1}=\frac{K-I+1}{I}-\frac{1}{{K \choose I-1}} 
		> \frac{K-I}{I+1}+\left(\frac{1}{I}-\frac{1}{{K \choose I-1}}\right).	
	\end{equation} 
	As $c_2^I=\frac{K-I}{I+1}$ and $I \le  {K \choose I-1}$ for $I >1$, we have $c_1^{I-1} > c_2^I$. 
	

	Similarly, we need to prove  $d_1^{I-1} < d_2^I$ to guarantee the consistence when combining ${\boldsymbol d_1}$ and ${\boldsymbol d_2}$ which is required in the third step. Hence, we derive 
	\begin{align}\label{breakp2}
		d_1^{I-1}-d_2^I&=(c_2^I-c_1^{I-1})-d_2^I \notag\\
		&=\frac{(K-I+1)(K-I)\dots 4}{K(K-1)\dots (I+3)}\times \frac{3 \times 2}{(I+2)(I+1)I} 
		-\frac{2(K+1)}{I(I+1)(I+2)}. 
	\end{align} 
	According to \eqref{breakp2}, it holds true that $d_1^{I-1}-d_2^I < 0$ when $K \ge 4$. Because there are at least three points, i.e. $I-1, I, I+1$, are considered, it implies that $I+1 \le K-1$. If more than one inactive users are targeted, it follows that $K \ge 4$.  
	
	\subsubsection{The second derivative of coefficients ${\boldsymbol c}$} 
	The second derivative of ${\boldsymbol c_1}$ defined as ${\boldsymbol e_1}=[e_1^0,e_1^1,\dots, e_1^{I-3}]$ is  
	\begin{align}
		e_1^l&=d_1^{l+1}-d_1^l=c_1^{l+2}-2c_1^{l+1}+c_1^l \notag\\
		&=\frac{2(K+1)}{(l+3)(l+2)(l+1)}-\frac{{I \choose l+3}}{{K \choose l+2}}+2\frac{{I \choose l+2}}{{K \choose l+1}}-\frac{{I \choose l+1}}{{K \choose l}}.
		\label{e1}
	\end{align} 
	For sake of simplification, we introduce a new group of variables $\boldsymbol{\hat{e}_1}=\{\hat{e}_1^l\}$ with $\hat{e}_1^l=(l+3)(l+2)(l+1) K(K-1)\dots(I+1) e_1^l, l \in [0:I-1]$. $\boldsymbol{\hat{e}_1}$ can be written as 
	\begin{align} \label{hat-e1}
		\hat{e}_1^l  
		=&2(K+1)K \dots(I+1)-(K-l-2)(K-l-3)\dots(I-l) \notag \\ 
		&\times \left[(I-l-1)(I-l-2)(l+2)(l+1) \right.
		\left. +(K-l)(K-l-1)(l+3)(l+2)\right. \\
		&\left. -2(K-l-1)(I-l-1)(l+3)(l+1) \right]. \notag 
	\end{align}  
	
	Firstly, some hypothesis is given based on numerical results: 
	\begin{itemize} 
		\item When $l$ is fixed, $\hat{e}_1^l(I)$\footnote{Here we utilize $\hat{e}_1^l(\alpha)$ to specify $\hat{e}_1^l$ with $I$ fixed at $I=\alpha$.} decreases with the increase of $I$; 
		\item When $I$ is fixed, $\hat{e}_1^l(I)$ is increasing with regarding to $l$;  
		\item Based on the previous properties, one can derive that the minimum value of $\hat{e}_1^l(I), ~l \in [0:I-1],~I \in [K-1] $ is $\hat{e}_1^0$ with $I=K-1$, i.e. $\hat{e}_1^0(K-1)$. 
	\end{itemize}
	If one can verify the hypotheses one by one, $\hat{e}_1^l \ge 0$ is then proved. To simplify the proof, we shall relax the conditions because when the first hypothesis is proved, we only need to focus on $I=K-1$ and prove that $\hat{e}_1^l$ is an increase function with regarding to $l$ when $I=K-1$ (or a constant function in particular). Then it will holds true that $\min_{l, I}~{\hat{e}_1^l}(I) \ge 0$, so that $\hat{e}_1^l(I) \ge 0, l \in [0:I-1], I \in [K-1]$. 
	
	Firstly, we shall prove that $\hat{e}_1^l(I+1) < \hat{e}_1^l(I)$ for any $l \le I$.  
	For sake of clarification, here we introduce three auxiliary variables $A, B, C_{I}$ given by  
	\begin{align}
		\left\{\!\!\begin{array}{cl}
			&A=2(K+1)K \dots(I+2), \\
			&B=(K-l-2)(K-l-3)\dots (I-l+1), \\
			&C_{I}=(I-l-1)(I-l-2)(l+2)(l+1) \\
			&~~~~~+(K-l)(K-l-1)(l+3)(l+2)\\
			&~~~~~-2(K-l-1)(I-l-1)(l+3)(l+1).   
		\end{array} \right. 
	\end{align}
	Using the expressions above, we obtain that
	\begin{align}\label{delta-hat-e1} 
		\hat{e}_1^l(I)-\hat{e}_1^l(I+1) 
		&=A(I+1)-B(I-l)(C_{I})-A-B(C_{I+1})  \notag\\
		&=AI-B\left((I-l)C_{I}-C_{I+1}\right)  \\
		&=B\left[2(K+1)KI \gamma-((I-l)C_{I}-C_{I+1})\right] \notag
	\end{align} 
   where $\gamma=\frac{A}{2(K+1)KB}=\frac{\alpha}{\beta}$ with auxiliary variables defined as $\alpha=(K-1)\dots (I+2)$ and $\beta=(K-l-2)\dots (I-l+1)$.
	
	Now \eqref{delta-hat-e1} simplifies into $\hat{e}_1^l(I)-\hat{e}_1^l(I+1) =B \times \Xi,$  
	with the expression $\Xi$ defined as 
	\begin{align}\label{Lambda-hat-e1} 	
		\Xi=&2(K+1)KI \gamma  
		-(I-l)(I-l-1)(l+1)(l+2)(I-l-3)~~ \notag\\
		&-(K-l)(K-l-1)(l+3)(l+2)(I-l-1) \\
		&+2(K-l-1)(I-l)(l+3)(l+1)(I-l-2).\notag
	\end{align} 
	In this case, it remains to prove that $\Xi > 0$ to guarantee $\hat{e}_1^l(I) > \hat{e}_1^l(I+1)$. 
	With $K, I$ and $l$ defined above, i.e. $0 \le l \le I-1, I \le K-2, K \in [3: +\infty]$, it holds true that $\Xi > 0$. This can be easily proved based on calculation so that the detailed derivation is omitted here for briefness. The main idea is to compute the derivation of $\Xi$ in terms of $K$ utilizing the derivation of $\gamma$ and then prove it is positive, so that one can easily get the minimum value of $\Xi$ which is positive ($\min{\Xi}=12$), with the smallest $K=3$ and the unique combination $(I, l)=(1, 0)$.  

	Second, we analyze the values of $\boldsymbol{\hat{e}_1}$ when $I=K-1$ which have been proved to be minimum values with regarding to $I$. 
	
	In this case, the minuend term in \eqref{hat-e1} is independent on $l$, so that the subtrahend referring to as $\Psi^l$ is the only term requiring to be analyzed. We then rewrite \eqref{hat-e1} into 
	\begin{equation}
		\hat{e}_1^l(K-1)=2(K+1)K -\Psi^l,
	\end{equation}
	where variable $\Psi^l$ defined as \footnote{Note that the coefficient term $(K-l-2)(K-l-3)\dots(I-l)$ in \eqref{psi} equals to $1$ as it is derived sequentially on condition of $(K-l-2)<(K-l-3)\dots<(I-l)$.}
	\begin{align}
		\Psi^l=&(K-l-2)(K-l-3)(l+2)(l+1)  
		+(K-l)(K-l-1)(l+3)(l+2) \notag\\
		  &-2(K-l-1)(K-l-2)(l+3)(l+1). \label{psi}
	\end{align} 
	The aim is to investigate how $\Psi^l$ varies with $l$ to find the minimum value of $\Psi^l$. Hence, we reformulate \eqref{psi} into 
	\begin{equation}
		 \Psi^l=(K-l-1)(l+3)(K+l+2)-(K-l-2)(l+1)(K+l+3)  
		=2(K+1)K. \label{psi1}
	\end{equation}    
	Now we can obtain that 
	\begin{equation}
		\hat{e}_1^l(K-1)=2(K+1)K -\Psi^l  
		=2(K+1)K-2(K+1)K=0,
	\end{equation}
	which states that $\hat{e}_1^l(K-1)$ is constant and always non-negative. As is proved previously that when $l$ is fixed, $\hat{e}_1^l(I)$ decreases with the increase of $I$, it holds true that 
	\begin{equation}
		\min_{l, I}~{\hat{e}_1^l}(I)=\min_l {\hat{e}_1^l}(K-1) \ge 0,~\text{if}~K \ge I+1.
	\end{equation}
	
	To summarize, it has been proved that $\hat{e}_1^l(I) \ge 0,~l \in [0:I-1],~I \in [K-1],~K \in [3:+\infty]$, which guarantees $e_1^l \ge 0$ in the considered interval. To be exact, we obtain $\hat{e}_1^l(I) > \hat{e}_1^l(K-1) = 0,~I \in [K-2]$. Hence, the second derivative of the coefficients is positive for both $e_1^l$ and $e_2^l$ when $I \in [K-2]$ with an exception of constant second derivative at $0$ when $I=K-1$. This ends the proof of Lemma~\ref{lemma-step1}. 
	
	\subsection{Proof of Lemma~\ref{deltaGstart}}
	To proceed, we reformulate $\Delta_{\tilde{G}_1}(t+1)$ into
	\begin{equation}
		\Delta_{\tilde{G}_1}(t+1)={t+1 \choose t}/{K \choose t}=\frac{(K-t)\dots2}{K\dots(t+2)}. 
	\end{equation}
	Apparently, $\Delta_{\tilde{G}_1}(t+1)$ is subject to $t$ but is independent on $I$. Then we reformulate $\Delta_{\tilde{G}_1}(t+1)$ into a function of $t$ defined as $\omega_1(t)=\Delta_{\tilde{G}_1}(t+1)$. In consideration of different values of $t$, it can be obtained that 
	\begin{align*} 
		&\omega_1(t)=\Delta_{\tilde{G}_1}(t+1)=\frac{(K-t)\dots2}{K\dots(t+2)}, \\  
		&\omega_1(t+1)=\Delta_{\tilde{G}_1}(t+2)=\frac{(K-t-1)\dots2}{K\dots(t+3)}, \\
		&\Omega_1(t)=\frac{\omega_1(t+1)}{\omega_1(t)}=-1+\frac{K+2}{K-t},
	\end{align*}
	where $\Omega_1(t)$ denotes the first derivative of $\omega_1(t)$ with regarding to $t$ in division form when $t \in [1:K-2]$. Apparently, $\Omega_1(t)$ increases with regarding to $t$. Let $t^*$ satisfy $\Omega_1(t^*)=1$ which implies that $\omega_1(t^*)=\omega_1(t^*+1)$. It then follows that $t^*=\floor{K/2}-1$. Consequently, it holds true that
	\begin{align}\label{deltaGOmega}
		\Omega_1(t) \left\{\begin{array}{cl}
			\le 1, & \mbox{if}~~1 \le t \le \floor{K/2}-1, \\
			\ge 1, & \mbox{if}~~ \ceil{K/2}-1 \le t \le K-2.
		\end{array}\right. 
	\end{align} 
	Hence, $\omega_1(t)$ firstly decreases with $t$ when $1 \le t \le \floor{K/2}-1$, and then increases when $\ceil{K/2} \le t \le K-2$. The maximum value of $\omega_1(t)$ is given by 
	\begin{equation}
		\max_t~\omega_1(t)=\max\{\omega_1(1), \omega_1(K-2)\}=\frac{2}{K}, 
	\end{equation}
	where $\omega_1(1)=\omega_1(K-2)=2/K$.
	
	Next, we focus on how $\Delta_{\tilde{G}_2}(t+1)$ varies with $t$. Similarly, it is derived that 
	\begin{align*} 
		&\omega_2(t)=\frac{1}{\left(t+1\right)\left(1+(K-t-2)t/K\right)}, \\
		&\omega_2(t+1)=\frac{1}{\left(t+2\right)\left(1+(K-t-3)(t+1)/K\right)}, \\
		&\Omega_2(t)=\frac{\omega_2(t+1)}{\omega_2(t)}=\frac{\left(t+1\right)\left(1+(K-t-2)t/K\right)}{\left(t+2\right)\left(1+(K-t-3)(t+1)/K\right)} \notag \\ 
		&~~~~~~=\frac{A}{A+B}, 
        ~~~~t \in [1:K-2],~K \ge 3, 
	\end{align*}
	where the definitions of $\omega_2(t)$ and $\Omega_2(t)$ are similar to $\omega_1(t)$ and $\Omega_1(t)$. The auxiliary variables are defined as 
	\begin{align*}
     	A&=\left(t+1\right)\left(K+(K-t-2)t\right), \\
		B&=-3t^2+(2K-9)t+3K-6.   
	\end{align*}
	The sign of $B$ is important. Since $B$ is a typical quadratic equation of $t$, the positive solution to $B=0$ is  
	\begin{equation*}
		t_2=\frac{2K-9+\sqrt{4K^2+9}}{6}  
		  \le K-2.  
	\end{equation*}
	Thus, we derive $B \ge 0$, if $1 \le t \le t_2$, and $B< 0$, if $t_2 < t \le K-2$.
	As a result, $\Omega_2(t)$ becomes 
	\begin{align} 
		\Omega_2(t) \left\{\begin{array}{cl}
			\le 1, & \mbox{if}~~1 \le t \le t_2,\\
			> 1, & \mbox{if}~~ t_2 < t \le K-2,
		\end{array}\right. 
	\end{align} 
	which means that $\omega_2(t)$ decreases with regarding to $t$ when $1 \le t \le t_2$ and then increases when $t_2 < t \le K-2$. 
	
	Now we compare $\omega_1(K-2)$ and $\omega_2(K-2)$ with $t=K-2$. As mentioned before, $\omega_1(1)=\omega_1(K-2)=2/K$.  
	\begin{equation*}  	
		\omega_2(K-2)=\frac{1-M/N}{1+M/N} 
		=\frac{1}{K-1} > \omega_1(K-2).
	\end{equation*}
	On the contrary, when $t=1$, we obtain
	\begin{equation*}  	
		\omega_2(1)=\frac{1}{2(1+(K-3)/K)}=\frac{1}{2(2-3/K)} \in [\frac{1}{4}, \frac{1}{2}].
	\end{equation*} 
	Comparing $\omega_1(1)$ and $\omega_2(1)$, we let $\omega_2(1)=\omega_1(1)$ and obtain  
	\begin{equation*}  	
	 (K-2)(K-6)=0.
	\end{equation*}   
	Therefore, it is proved that 
	\begin{align} \label{deltaG2}
		\left\{\begin{array}{cl}
			\Delta_{\tilde{G}}(2)=\omega_2(1)-\omega_1(1)<0, & \mbox{if}~~K \in [3:6),\\
			\Delta_{\tilde{G}}(2)=\omega_2(1)-\omega_1(1)\ge 0, & \mbox{if}~~K \in [6:+\infty].
		\end{array}\right. 
	\end{align}
	
	For the point next to the maximum $t$ ($t=K-3$), we get
	\begin{align*}  	
		\omega_2(K-3)&=\frac{1}{(K-2)(1+\frac{K-3}{K})},\\
		\omega_1(K-3)&=\frac{6}{K(K-1)}.
	\end{align*}
	Similarly by comparison, we obtain  
	\begin{align} 
		\left\{\begin{array}{cl}
			\omega_2(K-3)<\omega_1(K-3), & \mbox{if}~~3 \le K \le 8,\\
			\omega_2(K-3)>\omega_1(K-3), & \mbox{if}~~ K \ge 9.
		\end{array}\right. 
	\end{align}
    Consequently, $\Delta_{\tilde{G}}(t+1)$ at $(t=K-3)$ becomes
	\begin{align}\label{deltaGK-3} 
		\Delta_{\tilde{G}}(K-2)\left\{\begin{array}{cl}
			<0, & \mbox{if}~~3 \le K \le 8,\\
			>0, & \mbox{if}~~ K \ge 9.
		\end{array}\right. 
	\end{align}
	The conclusion on $\Delta_{\tilde{G}}(t+1)$ with $t=1$ and $t=K-3$, which are given in \eqref{deltaG2} and \eqref{deltaGK-3} can be utilized to analyze the value of $\Delta_{\tilde{G}}(t+1)$ when $t=1, \dots, K-4, K \ge 5$. In particular, for $t \in [5:8]$, it is easy to obtain that 
	\begin{align}\label{deltaGK-5} 
		\Delta_{\tilde{G}}(t+1) \left\{\begin{array}{cl}
			>0, & \mbox{if}~~t \in [K-4], K \in [6:8], \\
			<0, & \mbox{if}~~t \in [K-4], K=5.
		\end{array}\right. 
	\end{align}
	When $K \in [9:+\infty]$, the sign of $\Delta_{\tilde{G}}(t+1)$ with $t=K-3$ has been clarified above as $\Delta_{\tilde{G}}(t+1)>0$ for $t=K-3$. We can assume that there is some $t$ satisfying $\Delta_{\tilde{G}}(t+2)>0$, i.e. $\omega_2(t+1) > \omega_1(t+1)$. Then, $\Delta_{\tilde{G}}(t+1)$ simplifies to 
	\begin{align*} 
		\Delta_{\tilde{G}}(t+1)  
		=&\Delta_{\tilde{G}_2}(t+1)-\Delta_{\tilde{G}_1}(t+1) \\
		=&\omega_2(t)-\omega_1(t) 
		=\frac{\left[\omega_2(t+1)\Omega_1(t)-\omega_1(t)\Omega_2(t)\right]}{\Omega_1(t)\Omega_2(t)}.
	\end{align*} 
	Next, we compare $\Omega_1(t)(>0)$ and $\Omega_2(t)(>0)$ as 
	\begin{equation*}
		\Omega_1(t)-\Omega_2(t)  
		=\frac{t+2}{K-t}-\frac{\left(t+1\right)\left(K+(K-t-2)t\right)}{\left(t+2\right)\left(K+(K-t-3)(t+1)\right)}. 
	\end{equation*}
	It can further simplify to
	\begin{align*} 
		&(t+2)^2\left(K+(K-t-3)(t+1)\right) - (K-t)(t+1)\left(K+(K-t-2)t\right) \notag \\ 
		=&AB-CD,  
	\end{align*}
	where auxiliary variables $A=t^2+4t+4$, $B=K-2t-3$, $C=Kt-t^2+K-t$ and $D=Kt+K-2t^2-6t-4$, respectively. It follows that  
	\begin{align*}  	
		C&=A+(t+1)(K-2t-4)+t, \\
		D&=B+t(K-2t-4)-1. 
	\end{align*} 
	When $t>\ceil{K/2}-2$, $C<A, D<B$ holds true. Since integer $t$ is assumed, we obtain that $(K-2t-4) \le -1$ and $(t+1)(K-2t-4)+t<0$. It then holds true that $AB-CD >0$ and consequently $\Omega_1(t)>\Omega_2(t)$. Substituting this back and performing recursion from $t=K-4$, we obtain 
	\begin{equation*}
		\Delta_{\tilde{G}}(t+1)>0,~~~~\mbox{if}~t \in (\ceil{K/2}-2, K-4],~K>9.   
	\end{equation*} 
	Similarly, it can be proved that $\Delta_{\tilde{G}}(t+1)>0$ for $t \in [\ceil{K/2}-2], K>9$ utilizing recursion from $\Delta_{\tilde{G}}(2)=\omega_2(1)-\omega_2(1)>0$. In this case, we assume there is some $t$ satisfying $\Delta_{\tilde{G}}(t+1)>0$, then discuss the sign of $\Delta_{\tilde{G}}(t+2)$ given by
	\begin{equation*}
		\Delta_{\tilde{G}}(t+2)
		=\omega_2(t)\Omega_2(t)-\omega_1(t)\Omega_1(t).
	\end{equation*}
	Using the same strategy, $\Delta_{\tilde{G}}(t+1)>0$ for $t \in [\ceil{K/2}-2], K>9$ can be proved as well.  
   This ends the discussion on the sign of $\Delta_{\tilde{G}}(t+1)$ for all $t \in [K-2]$ and $K>3$.

	
	\bibliographystyle{IEEEtran}	
	\bibliography{CodedCachingUserInactivityTrans}

\end{document}